\documentclass[11pt,twoside]{article}
\pdfoutput=1

\usepackage{pslatex}    
\usepackage{amsfonts}
\usepackage{amsmath}
\usepackage{amsthm}
\usepackage{amscd}
\usepackage{cite}
\usepackage{eucal}
\usepackage{graphicx}
\usepackage{a4}
\usepackage{booktabs}
\usepackage{microtype}
\usepackage{appendix}
\usepackage{float}

\def\beq{\begin{equation}}
\def\eeq{\end{equation}}
\def\bea{\begin{eqnarray}}
\def\eea{\end{eqnarray}}

\def\beann{\begin{eqnarray*}}
\def\eeann{\end{eqnarray*}}

\let\a=\alpha \let\be=\beta \let\g=\gamma \let\de=\delta
\let\e=\varepsilon  \let\h=\eta \let\th=\theta
\let\eps=\epsilon
 \let\k=\kappa \let\la=\lambda \let\m=\mu
\let\n=\nu \let\x=\xi \let\p=\pi \let\r=\rho \let\s=\sigma
 
 \let\Ph=\phi  
  
 \let\G=\Gamma \let\D=\Delta

\let\qd=\quad  

\def\epp{\, .}
\def\epc{\, ,}

\theoremstyle{plain}

\newtheorem{lemma}{Lemma}

\newtheorem{corollary}{Corollary}
\newtheorem*{corollary*}{Corollary}

\newtheorem*{conjecture*}{Conjecture}

\theoremstyle{definition}

\def\2{\frac{1}{2}} \def\4{\frac{1}{4}}

\def\6{\partial}

\def\+{\dagger}

\def\<{\langle} \def\>{\rangle}

\let\then=\Rightarrow

\def\CO{{\cal O}}

\def\i{{\rm i}}

\def\rd{{\rm d}}
\def\re{{\rm e}}

\DeclareMathOperator{\sh}{sh}
\DeclareMathOperator{\ch}{ch}
\DeclareMathOperator{\tgh}{th}

\DeclareMathOperator{\cth}{cth}

\DeclareMathOperator{\Tr}{Tr}

\DeclareMathOperator{\sign}{sign}
\DeclareMathOperator{\End}{End}

\def\Re{{\rm Re\,}} \def\Im{{\rm Im\,}}

\def\fa{\mathfrak{a}}

\newcommand\uq{\overline{u}}

\renewcommand{\appendix}{%
   \renewcommand{\section}{
        \secdef\Appendix\sAppendix}%
   \setcounter{section}{0}%
   \renewcommand{\thesection}{\Alph{section}}%
   \renewcommand{\theequation}{\thesection.\arabic{equation}}%
}

\newcommand{\Appendix}[2][?]{%
     \refstepcounter{section}%
     \setcounter{equation}{0}%
     \addcontentsline{toc}{appendix}%
          {\protect\numberline{\appendixname~\thesection} #1}%
     \vspace{\baselineskip}%
     {\noindent\large\bfseries\appendixname\ \thesection: #2\par}%
     \sectionmark{#1}\vspace{\baselineskip}}

\newcommand{\sAppendix}[1]{%
     {\noindent\large\bfseries\appendixname\:: #1\par}%
     \sectionmark{#1}\vspace{\baselineskip}}

\begin{document}

\thispagestyle{empty}

\begin{center}

{\Large {\bf Low-temperature large-distance asymptotics of the transversal
two-point functions of the XXZ chain
\\}}

\vspace{7mm}

{\large
Maxime Dugave\footnote{e-mail: dugave@uni-wuppertal.de},
Frank G\"{o}hmann\footnote{e-mail: goehmann@uni-wuppertal.de}}%
\\[1ex]
Fachbereich C -- Physik, Bergische Universit\"at Wuppertal,\\
42097 Wuppertal, Germany\\[2.5ex]
{\large Karol K. Kozlowski\footnote{e-mail: karol.kozlowski@u-bourgogne.fr}}%
\\[1ex]
IMB, UMR 5584 du CNRS,
Universit\'e de Bourgogne, France

\vspace{40mm}

{\large {\bf Abstract}}

\end{center}

\begin{list}{}{\addtolength{\rightmargin}{9mm}
               \addtolength{\topsep}{-5mm}}
\item
We derive the low-temperature large-distance asymptotics
of the transversal two-point functions of the XXZ chain
by summing up the asymptotically dominant terms of their
expansion into form factors of the quantum transfer matrix.
Our asymptotic formulae are numerically efficient and
match well with known results for vanishing magnetic
field and for short distances and magnetic fields below
the saturation field.
\\[2ex]
\end{list}

\clearpage

\section{Introduction}
Many-body quantum systems undergoing second order phase
transitions renormalize to conformally invariant quantum
field theories at their critical points. The form of their
two-point functions at critical points is therefore
determined by the conformal group \cite{Polyakov70}.
These statements, which are at the heart of our understanding
of many-body quantum statistical systems are hard to
verify directly by analytic means, even for integrable
models. For the paradigmatic spin-$1/2$ anisotropic
Heisenberg chain a direct calculation of the large-distance
asymptotics of its two-point ground state correlation functions
was carried out only rather recently by two different methods
\cite{KKMST09a,KKMST11b}.

The anisotropic Heisenberg chain, or XXZ chain, is a model
of spins interacting with their nearest neighbours by
a linear combination of exchange and Ising interactions.
For local spins $1/2$ its Hamiltonian
\begin{equation} \label{ham}
     H_L = J \sum_{j = - L + 1}^L \Bigl( \s_{j-1}^x \s_j^x + \s_{j-1}^y \s_j^y
                               + \D \bigl( \s_{j-1}^z \s_j^z - 1 \bigr) \Bigr)
\end{equation}
is integrable. Here the $\s_j^\a$, $\a = x, y, z$, are the
Pauli matrices embedded into $\End \bigl( {\mathbb C}_2^{\otimes 2L} \bigr)$
and periodic boundary conditions, $\s_{-L}^\a = \s_L^\a$,
are implied. $J$ stands for the exchange energy and $\D$ for the
anisotropy parameter.

The first of the two above mentioned methods \cite{KKMST09a} is
based on the asymptotic analysis of a multiple integral
representation of a generating function of the longitudinal
two-point functions \cite{JiMi96,KMST05a}, while the second one
\cite{KKMST11b} relies on the analysis of form factors
\cite{KMT99a,KKMST09b,KKMST11a} (matrix elements of local
operators between the ground state and excited states of
the Hamiltonian). A major difficulty that had to be overcome
in the form factor approach comes from the fact that the
individual form factors of a system at a critical point
vanish in the thermodynamic limit. A possible way to cope
with this problem is to keep control over the finite-size
scaling behaviour of the form factors and to sum up the
leading finite-size contributions to the form factor
expansion before taking the thermodynamic limit. It was 
observed in \cite{KKMST11b} that such type of summation can
be performed by means of a combinatorial formula that arose
in the context of the representation theory of the infinite
symmetric group \cite{KOV93}.

In our previous work \cite{DGK13a} we combined the form
factor approach with the quantum transfer matrix approach
to the thermodynamics of integrable models \cite{Suzuki85,%
SuIn87,Kluemper93}. We calculated form
factors of the quantum transfer matrix that appear in the
form factor expansion of the thermal two-point correlation
functions of the XXZ chain. We like to call these form factors
thermal form factors. For finite temperatures the form
factor expansion of the two-point functions is a large-distance
asymptotic series expansion. Each term in the series consists
of an amplitude times a factor describing exponential decay.
The latter is characterized by a correlation length determined
by the ratio of an eigenvalue belonging to an excited state
of the quantum transfer matrix and the dominant eigenvalue.
The amplitudes are products of two thermal form factors.
For sufficiently high temperatures a few leading terms in
the form factor expansion are expected to determine the
correlation functions accurately. We plan to study this
intermediate and high-temperature regime in a separate work.

Here we continue the exploration of the low-temperature
regime, begun in \cite{DGK13a}, with the analysis of the
transversal correlation functions. In \cite{DGK13a} we studied
the low-temperature behaviour of the longitudinal correlation
functions in the antiferromagnetic critical regime for $|\D| < 1$.
In this regime the gap between the dominant eigenvalue and the
lowest excited-state eigenvalues of the quantum transfer matrix closes
in the zero-temperature limit $T \rightarrow 0$, and infinitely
many correlation lengths diverge. As a consequence of the overall
finiteness of the correlation functions the corresponding amplitudes
and form factors must necessarily vanish. The amplitudes 
decay as powers of $T$ with generally non-integer exponents whose
values depend on $\D$ and on the magnetic field in $z$-direction $h$.
Like in case of the usual transfer matrix these `critical' amplitudes
and correlation lengths can be classified by particle-hole
excitations above two Fermi points. The corresponding contributions
to the form factor series can be summed up by means of
the same summation formula as in case of the ordinary transfer
matrix \cite{KKMST11b}.

We observed this in our analysis of the low-temperature large-distance
asymptotics of the longitudinal two-point functions \cite{DGK13a}. 
In this work we carry out a similar low-temperature analysis of
the transversal two-point functions. We take the opportunity to
close some gaps in the arguments of our previous work. In particular,
we provide a careful and more complete discussion of the low-temperature
behaviour of the solutions of the non-linear integral equations
\cite{Kluemper92,Kluemper93} that are fundamental for the quantum
transfer matrix approach to the thermodynamics of integrable models.
We also show that the final formulae for the asymptotics of the
two-point functions are efficient. Our expressions for the leading
amplitudes can easily be evaluated numerically. Our formulae hold
for any finite magnetic field. In the limit of vanishing field we
compare them with the amplitudes for $h = 0$ obtained by Lukyanov
\cite{Lukyanov98,Lukyanov99} using field theoretical methods.
We also find amazing agreement, when we compare the low-temperature
magnetic field dependence of the third neighbour correlation functions
following from our large-distance asymptotic expansion with the known
exact results \cite{BDGKSW08}.

\section{Form factor expansion of two-point functions}
The thermal expectation value of an operator $X$ acting on the
space of states of the XXZ chain is defined as
\begin{equation}
     \< X \> = \lim_{L \rightarrow \infty}
               \frac{\Tr \bigl\{\re^{- H_L/T + h S^z/T} X \bigr\}}
                    {\Tr \{\re^{- H_L/T + h S^z/T}\}} \epc
\end{equation}
where
\begin{equation}
     S^z = \2 \sum_{j = - L + 1}^L \s_j^z
\end{equation}
is the conserved $z$-component of the total spin.

In our previous work \cite{DGK13a} we considered the longitudinal
and transversal two-point functions $\< \s_1^z \s_{m+1}^z \>$ and
$\< \s_1^- \s_{m+1}^+ \>$. We studied the longitudinal two-point
function by means of a generating function. For the transversal
two-point function no convenient generating function is known.
We considered a direct form-factor expansion of the form
\begin{equation} \label{pmexp}
     \<\s_1^- \s_{m+1}^+\> = \sum_{n=1}^\infty A_n^{-+} \re^{- m/\x_n} \epp
\end{equation}
Every term in this series is characterized by a correlation length
$\x_n$ and by an amplitude~$A_n^{-+}$.

Correlation lengths had been studied in various contexts before
\cite{KlSc03,SSSU99,Takahashi91,Kluemper92,PaKl13}. They are related to
ratios $\r_n$ of quantum transfer matrix eigenvalues \cite{Suzuki85}
in the so-called Trotter limit,
\begin{equation}
     \re^{- 1/\x_n} = \r_n \epp
\end{equation}
Here $\r_n = \r_n (0|0)$ and
\begin{equation} \label{rhoint}
     \r_n (\la|\a) = q^{s + \a}
        \exp \biggl\{\int_{{\cal C}_n} \frac{\rd \m}{2\p \i} \: \re (\m - \la)
	\ln \biggl( \frac{1 + \fa_n (\m|\a)}{1 + \fa_0 (\m)} \biggr) \biggr\} \epp
\end{equation}

In this expression $q = \re^\h$ parameterizes the anisotropy parameter,
$\D = (q + q^{-1})/2$. In order to avoid case differentiations we shall
assume in the following that $\h = - \i \g$, $\g \in (0,\p/2)$. Then
$0 < \D < 1$. We shall comment on the case $- 1 < \D < 0$ below and
plan to treat the `massive case' $\D > 1$ in a separate work. The integer
$s$ refers to the $z$-component of the pseudo spin of the exited state
of the quantum transfer matrix that characterizes $\x_n$. For the
transversal two-point functions (\ref{pmexp}) it must be equal to
the eigenvalue belonging to $\s^+$ under the adjoint action of $\s^z/2$,
\textit{i.e.}\ $s = 1$.

The auxiliary functions $\fa_n (\cdot|\a)$ are solutions of nonlinear
integral equations,
\begin{equation} \label{nlie}
     \ln (\fa_n (\la|\a)) =  - 2 (\k + \a + s) \h - \be \re (\la)
        - \int_{{\cal C}_n} \frac{\rd \m}{2\p \i} \:
	  K(\la - \m) \ln \bigl(1 + \fa_n (\m|\a)\bigr) \epp
\end{equation}
As usual in integrable systems everything here is determined by one- and
two-particle data. The function
\begin{equation} \label{baree}
     \re (\la) = \cth (\la) - \cth(\la + \h)
\end{equation}
is the bare energy, and the kernel
\begin{equation} \label{kernel}
     K(\la) = \cth (\la - \h) - \cth (\la + \h)
\end{equation}
is the derivative of the bare two-particle scattering phase. The
twist $\a \in {\mathbb C}$ is an important regularization parameter
as will become clear below. It has a physical meaning in the
conformal limit \cite{BJMS10}. In equation (\ref{nlie}) it appears
as a shift of the rescaled magnetic field $\k$. This field and the
rescaled inverse temperature $\be$ are defined as
\begin{equation}
     \k = \frac{h}{2 \h T} \epc \qd \be = \frac{2J \sh(\h)}{T} \epp
\end{equation}
The physical correlation functions are obtained in the limit
$\a \rightarrow 0$. For this reason we may assume that
$0 < |\a| \ll 1$.

Solutions of (\ref{nlie}) are characterized by equivalence classes of
contours ${\cal C}_n$, two contours being equivalent if they admit
the same solution $\fa_n (\cdot|\a)$. The `canonical contour' ${\cal C}_0$
associated with the dominant eigenvalue of the quantum transfer matrix
consists of two straight lines parallel to the real axis and passing
through $\mp \i (\g/2 - 0^+)$. It is oriented in such a way that
it encircles the real axis in a counterclockwise manner. The corresponding
untwisted auxiliary function, the solution $\fa_0 = \fa_0 (\cdot|0)$ of
(\ref{nlie}) in the pseudo spin $s = 0$ sector, characterizes the thermal
equilibrium of the model. It parameterizes the free energy per lattice site
\cite{Kluemper93}, the multiple-integral representations of
temperature correlation functions \cite{GKS04a,GKS05} and the
physical part of the temperature correlation functions in factorized
form \cite{BoGo09}. For $s \ne 0$ there are reference contours
${\cal C}_{0,s}$ associated with the eigenvalue of largest modulus
in the respective sector. The case $s = 1$ will be described in
detail below.
\begin{figure}
\begin{center}
\includegraphics[width=.8\textwidth]{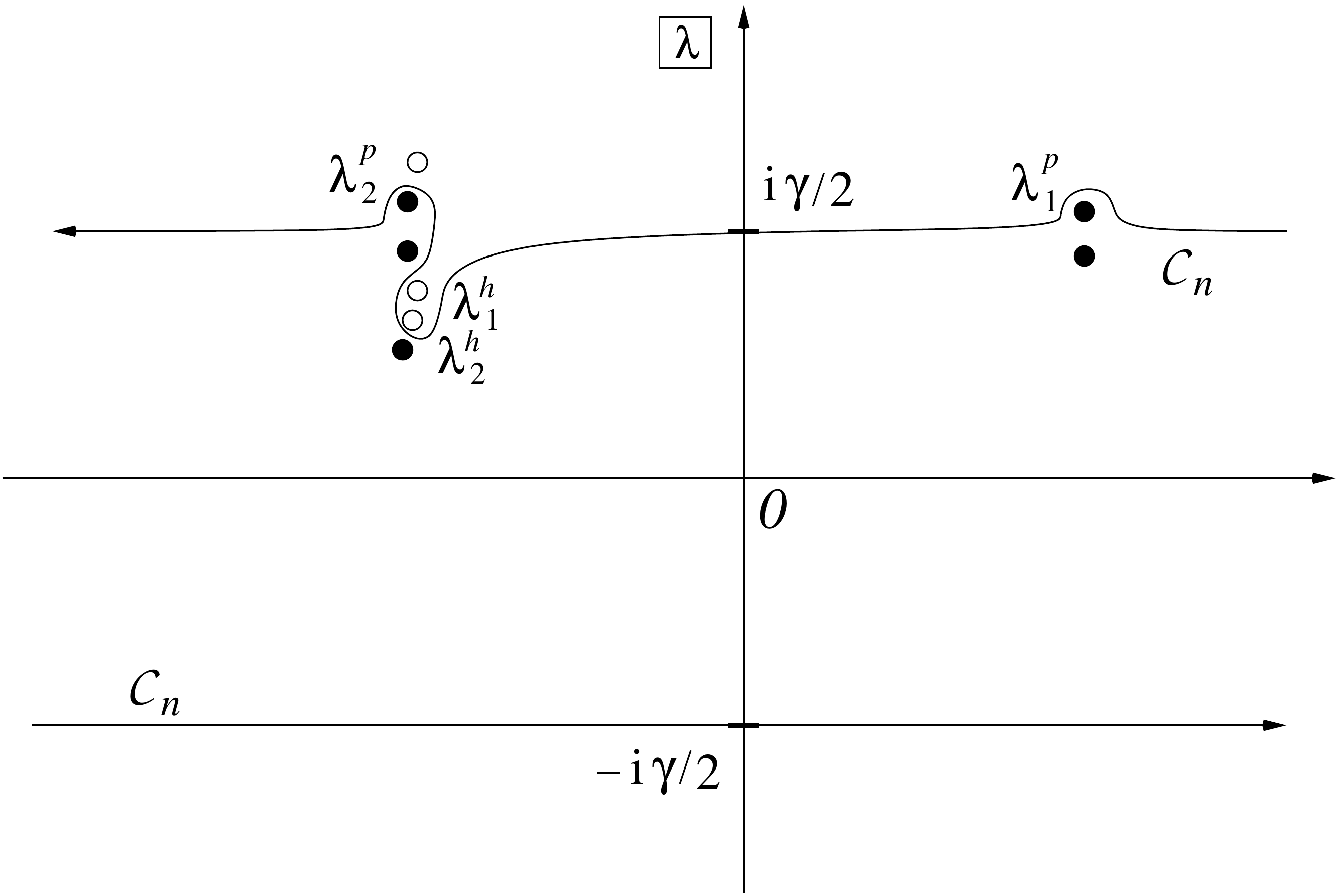}
\caption{\label{fig:cn} Sketch of a contour ${\cal C}_n$ for $0 < \D < 1$.
Holes $\la_j^h$ inside the `fundamental strip' between $- \i \g/2$ and
$\i \g/2$ are excluded from the contour, and particles $\la_k^p$ outside
the fundamental strip are included into the contour.
}
\end{center}
\end{figure}
The contours ${\cal C}_n$ are deformations of the reference contours.
Zeros $\la^h_j$ of $1 + \fa_n (\cdot|\a)$ inside ${\cal C}_{0,s}$
but outside ${\cal C}_n$ will be called holes, and zeros $\la_k^p$ of
$1 + \fa_n (\cdot|\a)$ outside ${\cal C}_{0,s}$ but inside ${\cal C}_n$
will be called particles.

Sets of zeros and holes and equivalence classes of contours are
in one-to-one correspondence. This can be seen by deforming
the contours in (\ref{nlie}) into the reference contours ${\cal C}_{0,s}$
(`straightening the contours'). Then the particles and holes appear
as explicit parameters in the driving terms of the integral equation
and must be determined by the subsidiary conditions $1 +
\fa_n (\la^{p, h}|\a) = 0$. This formulation will be useful in
section~\ref{sec:lowt} on the low-temperature limit.

Expressions for the amplitudes $A_n^{-+}$ in the form-factor
expansion of the transversal correlation functions were obtained
in \cite{DGK13a}, $A_n^{-+} = \lim_{\a \rightarrow 0}
A_n^{-+} (0|\a)$, where
\begin{multline} \label{amp}
     A_n^{-+} (\x|\a) = \frac{\overline{G}_+^- (\x) \overline{G}_-^+ (\x)}
                          {(q^{1 + \a} - q^{- 1 - \a})(q^\a - q^{-\a})} \\[1ex] \times
                     \exp \biggl\{ -
		        \int_{{\cal C}_n} \frac{\rd \la}{2 \p \i} \:
			\ln \bigl( \r_n (\la|\a) \bigr) \6_\la
			\ln \biggl( \frac{1 + \fa_n (\la|\a)}
			    {1 + \fa_0 (\la)} \biggr) \biggr\}\\[1ex] \times
		     \frac{\det_{\rd m^\a_+, {\cal C}_n}
		           \bigl\{ 1 - \widehat{K}_{1-\a} \bigr\}
		           \det_{\rd m^\a_-, {\cal C}_n}
		           \bigl\{ 1 - \widehat{K}_{1+\a} \bigr\}}
			  {\det_{\rd m^\a_0, {\cal C}_n}
			   \bigl\{ 1 - \widehat{K} \bigr\}
			   \det_{\rd m, {\cal C}_n}
			   \bigl\{ 1 - \widehat{K} \bigr\}} \epp
\end{multline}

Here, for $s = \pm$,
\begin{equation}
     \overline{G}_s^\pm (\x)
        = \lim_{\Re \la \rightarrow \pm \infty} \overline{G}_s (\la, \x)
\end{equation}
and $\overline{G}_\pm (\la, \x)$ is the solution of the linear
integral equation
\begin{multline} \label{liegbar}
     \overline{G}_\pm (\la, \x) = 
        - \cth(\la - \x) + q^{\a \mp 1} \r_n^{\pm 1} (\x|\a) \cth(\la - \x - \h) \\
	+ \int_{{\cal C}_n} \rd m_\pm^\a (\m) \overline{G}_\pm (\m, \x)
	  K_{\a \mp 1} (\m - \la)
\end{multline}
with deformed kernel
\begin{equation}
     K_\a (\la) = q^{- \a} \cth(\la - \h) - q^\a \cth(\la + \h) \epp
\end{equation}
Note that $\overline{G}_-^+$ is analytic in $\a$ and that
$\overline{G}_-^+|_{\a = 0} = 0$ which implies that the limit
$\a \rightarrow 0$ exists in (\ref{amp}). The vanishing of
$\overline{G}_-^+$ at $\a = 0$ came out rather indirectly in the
derivation of (\ref{amp}) in \cite{DGK13a}. We wish to point
out that it can be shown more directly using the technique
developed in \cite{BoGo10}.

The `measures' $\rd m^\a_\eps$, $\e = -, 0, +$ and $\rd m$ are defined by
\begin{subequations}
\begin{align}
     \rd m_-^\a (\la) &
        = \frac{\rd \la \: \r_n^{-1} (\la|\a)}{2 \p \i (1 + \fa_0 (\la))} \epc \qd
     \rd m_+^\a (\la)
        = \frac{\rd \la \: \r_n (\la|\a)}{2 \p \i (1 + \fa_n (\la|\a))} \epc \\
     \rd m (\la) & = \frac{\rd \la}{2 \p \i (1 + \fa_0 (\la))} \epc \qd
     \rd m_0^\a (\la) = \frac{\rd \la}{2 \p \i (1 + \fa_n (\la|\a))} \epp
\end{align}
\end{subequations}
The determinants in (\ref{amp}) are Fredholm determinants of integral
operators defined by the respective kernels and measures and by the
integration contours ${\cal C}_n$ (see \cite{DGK13a} for more details).

We would like to emphasize that the amplitudes are entirely described
in terms of functions that appeared earlier in the description of
the thermodynamic properties, the correlation length and the correlation
functions of the model. These are the auxiliary functions $\fa_n$
\cite{Kluemper92,Kluemper93}, the eigenvalue ratios $\r_n$ \cite{JMS08,BoGo09}
and the deformed kernel $K_\a$ \cite{BGKS07}. The parameter $s$ corresponds
to the spins of the operators $\s^\pm$ occurring in the transversal
correlation function (see \cite{DGK13a}).

\section{The low-temperature limit}
\label{sec:lowt}
Equations (\ref{rhoint}) and (\ref{amp}) for the correlation lengths
and amplitudes in the asymptotic series (\ref{pmexp}) are valid for all
positive temperatures and magnetic fields and, with an appropriate
adaption of the definition of the contours, also for all $\D > - 1$.
In the general case we do not expect any further simplification
of these formulae. Their further analysis will have to rely on
numerical calculations. In the low-temperature limit, for $|\D| < 1$
and $h$ below the critical field in particular, we expect universal
behaviour described by conformal field theory, and further simplification
should be possible.

A similar low-temperature analysis was carried out for a
generating function of the density-density correlation functions
of the Bose gas with delta-function interaction \cite{KMS11a,KMS11b}
and for a generating function of the longitudinal two-point functions
of the XXZ chain \cite{DGK13a}. We shall see that the latter analysis
basically carries over to the case of the transversal correlation
functions considered in this work. It consists of the following
steps. We first straighten the contours in the nonlinear integral
equations (\ref{nlie}) and in the expressions of the correlation
lengths (\ref{rhoint}) and amplitudes (\ref{amp}), making the
dependence on particle and hole positions explicit. Then we
perform a low-temperature analysis of the nonlinear integral
equations (\ref{nlie}) using a generalization of Sommerfeld's
lemma \cite{Sommerfeld28}. Note that this is only possible for
non-zero magnetic field as it relies on the existence of a pair
of Fermi points. The low-temperature data from the analysis of
the nonlinear integral equations then enter the analysis of the
correlation lengths and amplitudes. It turns out that infinitely
many correlation lengths diverge in the low-temperature limit
and that the corresponding amplitudes all vanish as $T$ goes to
zero. Therefore, using a formula derived in \cite{KOV93,Olshanski03,%
KKMST11b}, we first sum up the leading low-temperature
contributions which then allows us to perform the zero-temperature
limit. We shall take the opportunity to present some details and
proofs that were omitted in our previous work \cite{DGK13a}.
\subsection{Straightening the contours}
We denote the number of holes by $n_h$ and the number of particles
by $n_p$ and define the functions
\begin{equation}
     E(\la) = \ln \biggl( \frac{\sh (\la)}{\sh(\h + \la)} \biggr) \epc \qd
     \th (\la) = \ln \biggl( \frac{\sh (\h - \la)}{\sh(\h + \la)} \biggr)
\end{equation}
which are antiderivatives of the bare energy $\re (\la)$ and of
the kernel $K(\la)$. We further introduce the shorthand notation
\begin{equation} \label{defz}
     z(\la) = - \frac{1}{2 \p \i}
                \ln \biggl( \frac{1 + \fa_n (\la|\a)}{1 + \fa_0 (\la)} \biggr) \epp
\end{equation}

A straightening of the contours in (\ref{nlie}) leads to
\begin{multline} \label{nliecs}
     \ln (\fa_n (\la|\a)) =  \i \p s
        - 2 (\k + \a) \h - \be \re (\la) \\
        + \sum_{j=1}^{n_h} \th (\la - \la_j^h)
        - \sum_{j=1}^{n_p} \th (\la - \la_j^p)
        - \int_{{\cal C}_{0,s}} \frac{\rd \m}{2\p \i} \:
	  K(\la - \m) \ln \bigl(1 + \fa_n (\m|\a)\bigr) \epp
\end{multline}
Here the contour ${\cal C}_{0,s}$ is chosen in such a way that
\begin{equation} \label{npnhs}
     n_h - n_p = s \epp
\end{equation}
In the free fermion limit $\g \rightarrow \p/2$ and in the
low-temperature limit (see below) we shall see that ${\cal C}_{0,0}$
is equal to the canonical contour ${\cal C}_0$. We shall also
find an explicit description of ${\cal C}_{0,1}$.
Equation (\ref{nliecs}) defines an $(n_p + n_h)$-parametric family of
functions, depending on $\{\la_j^h\}$ and $\{\la_j^p\}$. The individual
functions $\fa_n (\cdot|\a)$ are then determined by the subsidiary conditions
\begin{equation} \label{subs}
     \fa_n (\la_j^h|\a) = \fa_n (\la_k^p|\a) = - 1 \epc \qd j = 1, \dots, n_h \epc \qd
                       k = 1, \dots, n_p
\end{equation}
fixing these parameters to a discrete set of values.

Straightening the contours the eigenvalue ratios take the form
\begin{equation} \label{evratc0}
     \r_n (\la|\a) = q^\a \exp \biggl\{
        \sum_{j=1}^{n_h} E (\la_j^h - \la ) - \sum_{j=1}^{n_p} E (\la_j^p - \la)
	- \int_{{\cal C}_{0,s}} \rd \m \: \re (\m - \la) z(\m) \biggr\}
\end{equation}
for $\la$ inside ${\cal C}_{0,s}$.

For the amplitudes we concentrate on the exponential term in (\ref{amp}).
Upon straightening the contours it turns into
\begin{align} \label{amps}
     A_n^{(0)} (\a) = 
     \exp \biggl\{ - \int_{{\cal C}_n} & \frac{\rd \la}{2 \p \i} \:
			\ln \bigl( \r_n (\la|\a) \bigr) \6_\la
			\ln \biggl( \frac{1 + \fa_n (\la|\a)}
			    {1 + \fa_0 (\la)} \biggr) \biggr\}\\[2ex]
     = & (-1)^s q^{s \a}
	         \frac{\prod_{j = 1}^{n_h} \prod_{k = 1}^{n_p}
	               \sh(\la_j^h - \la_k^p + \h) \sh(\la_j^h - \la_k^p - \h)}
                      {\Bigl[ \prod_{j, k = 1}^{n_h} \sh(\la_j^h - \la_k^h - \h) \Bigr]
                       \Bigl[ \prod_{j, k = 1}^{n_p} \sh(\la_j^p - \la_k^p + \h) \Bigr]}
		 \notag \\[.5ex]
        & \times 
	         \frac{\biggl[ \prod_{\substack{j, k = 1 \\ j \ne k}}^{n_h}
	                      \sh(\la_j^h - \la_k^h) \biggr]
                       \biggl[ \prod_{\substack{j, k = 1 \\ j \ne k}}^{n_p}
	                      \sh(\la_j^p - \la_k^p) \biggr]}
                      {\Bigl[ \prod_{j=1}^{n_h} \prod_{k=1}^{n_p}
		              \sh(\la_j^h - \la_k^p) \Bigr]^2}
			  \notag \\[.5ex]
        & \times 
	  \biggl[ \prod_{j=1}^{n_p}
	          \Bigl( \6_\la \re^{- 2 \p \i z(\la)} \Bigr)^{-1}_{\la = \la_j^p} \biggr]
	  \biggl[ \prod_{j=1}^{n_h}
	          \Bigl( \6_\la \re^{- 2 \p \i z(\la)} \Bigr)^{-1}_{\la = \la_j^h} \biggr]
		  \notag \displaybreak[0] \\[.5ex]
        & \times
	    \exp \biggl\{ \sum_{j=1}^{n_p} \int_{{\cal C}_{0,s}} \rd \la \: z(\la)
	                  \bigl( \re(\la - \la_j^p) - \re(\la_j^p - \la) \bigr) \biggr\}
			  \notag \\[.5ex]
        & \times
	    \exp \biggl\{ \sum_{j=1}^{n_h} \int_{{\cal C}_{0,s}} \rd \la \: z(\la)
	                  \bigl( \re(\la_j^h - \la) - \re(\la - \la_j^h) \bigr) \biggr\}
			      \notag \\[.5ex]
        & \times
	    \exp \biggl\{ - \int_{{\cal C}_{0,s}} \rd \la \: \int_{{\cal C}_{0,s}'} \rd \m \:
	         z(\la) \re' (\la - \m) z(\m) \biggr\}
		       \epp
\end{align}
Here ${\cal C}_{0,s}'$ is a contour infinitesimally close to ${\cal C}_{0,s}$ and
inside ${\cal C}_{0,s}$. The Fredholm determinants and the factors
$\overline{G}_\pm$ need a different treatment which will be discussed
below. In \cite{DGK13a} we called the above contribution to the amplitude
the `universal part', since it is of the same form in the longitudinal
and transversal case.

\subsection{Low-temperature analysis of the nonlinear integral equations}
In \cite{DGK13a} we performed a low-temperature analysis of the nonlinear
integral equations (\ref{nliecs}) that also suits for our present purposes.
Here we shall reproduce this analysis and take the opportunity to add
more details, to make our arguments more rigorous and to clarify the appearance 
and disappearance of certain phases `$p$'.

We introduce the notations
\begin{equation}
     \e_0 (\la) = h - \frac{4 J(1 - \D^2)}{\ch (2 \la) - \D}
\end{equation}
and
\begin{equation}
     u_0 (\la) = - T \ln \bigl( \fa_0 (\la + \i \g/2) \bigr) \epc \qd
     u(\la) = - T \ln \bigl( \fa_n (\la + \i \g/2 |\a) \bigr) \epc
\end{equation}
where we omitted the index $n$ and the dependence on $\a$ in the
definition of $u$. Then (\ref{nliecs}) implies that
\begin{multline} \label{nlieu}
     u (\la) = \e_0 (\la) + T \biggl[ 2 \p \i (\a' - s/2)
        + \sum_{j=1}^{n_p} \th (\la - \la_j^p + \i \g/2)
        - \sum_{j=1}^{n_h} \th (\la - \la_j^h + \i \g/2) \biggr] \\
        + T \int_{{\cal C}_{0,s} - \i \g/2} \frac{\rd \m}{2\p \i} \:
	  K(\la - \m) \ln \Bigl(1 + \re^{- \frac{u(\m)}T} \Bigr) \epc
\end{multline}
where $\a' = \h \a/\i \p$. A similar equation without the contribution
proportional to $T$ in the driving term holds for $u_0$,
\begin{equation} \label{nlieu0}
     u_0 (\la) = \e_0 (\la)
                 + T \int_{{\cal C}_0 - \i \g/2} \frac{\rd \m}{2\p \i} \:
		            K(\la - \m) \ln \Bigl(1 + \re^{- \frac{u_0(\m)}T} \Bigr) \epp
\end{equation}
Since $\th$ is bounded on the contour ${\cal C}_{0,s} - \i \g/2$, the terms in
square brackets in (\ref{nlieu}) may be neglected compared to $\e_0 (\la)$
when $T$ becomes small. Thus, $u$ and $u_0$ have the same zero temperature
limit $\e$.

Intuitively the zero-temperature limit of (\ref{nlieu}) and (\ref{nlieu0}) is
rather clear: the integrals in (\ref{nlieu}) and (\ref{nlieu0}) vanish on those
parts of the contour on which $\Re \e > 0$ and are nonzero on their complement.
From the behaviour of the driving term $\e_0 (\la)$ one may guess that this
complement is an interval $[-Q, Q]$ on the real axis. This can be stated more
precisely. For $0 < h < h_c = 4J(1 + \D)$ we define the dressed energy $\e$ as
the solution of the linear integral equation
\begin{equation} \label{lieeps}
     \e (\la) = \e_0 (\la) + \int_{-Q}^Q \frac{\rd \m}{2\p \i}
                \: K(\la - \m) \e (\m) \epc
\end{equation}
where $Q > 0$ as a function of $h$ is uniquely determined by the condition
$\e (Q) = 0$ (for a proof of the latter statement see \cite{DGK13bpp}).
$\e$ is real and even on $\mathbb R$ and monotonously increasing on
${\mathbb R}_+$. One can further prove \cite{DGK13bpp} that, for all
$\g \in (0,\p/2)$, $\Re \e > h/4 > 0$ on ${\mathbb R} - \i \g + \i 0$
which is the lower part of the integration contour in (\ref{nlieu}),
(\ref{nlieu0}).

Assuming uniqueness of the solutions of (\ref{nlieu}) and (\ref{nlieu0})
the above properties of $\e$ allow us to conclude that $\lim_{T \rightarrow 0}
u(\la) = \lim_{T \rightarrow 0} u_0 (\la) = \e (\la)$. This follows from the
following `generalized Sommerfeld lemma' which also allows us to obtain
the first and second order temperature corrections below.
\begin{lemma} \label{lem:sommerfeld}
Let $T > 0$. Let $u, f$ be holomorphic in an open set containing a contour
${\cal C}_u$, and f bounded on ${\cal C}_u$. Let $\ln \bigl(1 +
\re^{- u(\la)/T} \bigr)$ be continuous on ${\cal C}_u$ (this means
we consider ${\cal C}_u$ as a contour on the Riemann surface of
$\ln \bigl(1 + \re^{- u(\la)/T} \bigr)$ realized as a multi-sheeted
cover of the complex plane with cuts along the curves where the argument
of the logarithm is negative, say). Let $v = \Re u$, $w = \Im u$. Assume that
$v$ has exactly two zeros $Q_\pm$ on ${\cal C}_u$ dividing ${\cal C}_u$ into
two parts, ${\cal C}_u^-$ on which $v$ is negative and ${\cal C}_u^+$ on
which $v$ is positive. Let ${\cal C}_u$ be oriented in such a way that
$Q_-$ comes before $Q_+$ on ${\cal C}_u^-$. If there is a $p \in {\mathbb Z}$
such that $w(Q_\pm) = 2\p p T$ then there is a choice of branches of
$\ln \bigl(1 + \re^{- u(\la)/T} \bigr)$ for which
\begin{multline}
     T \int_{{\cal C}_u} \rd \la \: f(\la) \ln \Bigl(1 + \re^{- \frac{u(\la)}T} \Bigr)
        = - \int_{Q_-}^{Q_+} \rd \la \: f(\la) \bigl(u(\la) - 2 \p \i p T\bigr) \\
	  + \frac{T^2 \p^2}{6}
	    \biggl( \frac{f(Q_+)}{u'(Q_+)} - \frac{f(Q_-)}{u'(Q_-)} \biggr)
          + \CO (T^4) \epp
\end{multline}
\end{lemma}
\begin{proof}
Let $F$ be an antiderivative of $f$. Then 
\begin{multline}
     T \int_{{\cal C}_u^-} \rd \la \: f(\la) \ln \Bigl(1 + \re^{- \frac{u(\la)}T} \Bigr) \\
        = \int_{{\cal C}_u^-} \rd \la \:
	         \Bigl\{ T f(\la) \ln \Bigl(1 + \re^{\frac{u(\la)}T} \Bigr)
                 - f(\la) u(\la) + T \6_\la F(\la) g(\la) \Bigr\} \epc
\end{multline}
where
\begin{equation} \label{defg}
     g(\la) = \frac{u(\la)}T + \ln \Bigl(1 + \re^{- \frac{u(\la)}T} \Bigr)
                       - \ln \Bigl(1 + \re^{\frac{u(\la)}T} \Bigr) \epp
\end{equation}
This function is continuous on the contour ${\cal C}_u^-$ no matter how
we fix the branches of the logarithm, since the first two terms on the
right hand side are continuous by hypothesis and since $\Re \bigl(1 +
\re^{u(\la)/T} \bigr) \ge 0$ on ${\cal C}_u^-$. On the other hand
$\exp \bigl\{ g(\la) \bigr\} = 1$. $\then \exists n \in {\mathbb Z}$ such that
$g(\la) = 2 \p \i n$. Thus, no matter how we chose the branches of the
logarithms, there is always an $n \in {\mathbb Z}$ such that
\begin{multline} \label{switchlogpm}
     T \int_{{\cal C}_u^-} \rd \la \: f(\la) \ln \Bigl(1 + \re^{- \frac{u(\la)}T} \Bigr) \\
        = T \int_{{\cal C}_u^-} \rd \la \: f(\la) \ln \Bigl(1 + \re^{\frac{u(\la)}T} \Bigr)
          - \int_{{\cal C}_u^-} \rd \la \: f(\la) \bigl(u(\la) - 2 \p \i n T\bigr) \epp
\end{multline}

Our goal is to estimate the integral over ${\cal C}_u$ for small positive $T$.
If $\la \in {\cal C}_u^- \setminus \{Q_-, Q_+\}$ then $v(\la) < 0$ and
$\re^{u(\la)/T} = {\cal O} (T^\infty)$. Moreover, if we fix the branch of
$\ln \bigl(1 +\re^{u(Q_-)/T} \bigr)$ such that $\bigl| \arg \bigl(1 +
\re^{u(Q_-)/T} \bigr) \bigr| < \p/2\ \then \bigl| \arg \bigl(1 +
\re^{u(\la)/T} \bigr) \bigr| < \p/2$ for all $\la \in {\cal C}_u^-$,
and the first term on the right hand side of (\ref{switchlogpm}) has
no ${\cal O} (T)$ contribution. On any part of ${\cal C}_u$ that is
disconnected with $Q_-$ a similar argument applies. Since $v(\la) > 0$
on these parts, we may assume that $\bigl| \arg \bigl(1 +
\re^{- u(\la)/T} \bigr) \bigr| < \p/2$ and that there is no ${\cal O} (T)$
contribution. For the part of ${\cal C}_u^+$ that is connected with
$Q_-$ we chose the branch in such a way that $\bigl| \arg \bigl(1 +
\re^{- u(Q_-)/T} \bigr) \bigr| < \p/2$. Then there is again no 
${\cal O} (T)$ contribution. Now if
\begin{equation}
     u(Q_+) = u(Q_-) = 2 \p \i p T \epc
\end{equation}
for some $p \in {\mathbb Z}$ then (\ref{defg}) with $\la = Q_-$ implies
that $g(Q_-) = 2 \p \i n = 2 \pi \i p$. The logarithms cancel each other
because of our choice of branches. Using once more (\ref{defg}), this time
at $Q_+$, we conclude that
\begin{equation}
     \ln \Bigl( 1 + \re^{- \frac{u(Q_+)}T} \Bigr)
        = \ln \Bigl( 1 + \re^{\frac{u(Q_+)}T} \Bigr) \epp
\end{equation}
Thus, at $\la = Q_+$ we can continue $\ln \bigl(1 + \re^{u(\la)/T} \bigr)$
on ${\cal C}_u^-$ continuously into $\ln \bigl(1 + \re^{- u(\la)/T} \bigr)$
on ${\cal C}_u^+$, and $\bigl| \arg \bigl(1 + \re^{- u(Q_+)/T} \bigr) \bigr|
< \p/2$ which means that there is no ${\cal O} (T)$ contribution on
the part of ${\cal C}_u^+$ connected with $Q_+$ either.

It follows from (\ref{switchlogpm}) that
\begin{multline}
     \D I := T \int_{{\cal C}_u} \rd \la \:
                  f(\la) \ln \Bigl(1 + \re^{- \frac{u(\la)}T} \Bigr)
	    + \int_{Q_-}^{Q_+} \rd \la \: f(\la) \bigl(u(\la) - 2 \p \i p T\bigr) \\
          = T \int_{{\cal C}_u} \rd \la \:
	             f(\la) \ln \Bigl(1 + \re^{- \frac{u(\la) \sign (v(\la))}T} \Bigr)
		     \epc
\end{multline}
where the logarithm on the right hand side is continuous at $Q_\pm$ and
where the real part of the exponent on the right hand side is negative everywhere
except at $Q_\pm$. This means that the leading contribution to the integral for
$T \rightarrow 0$ comes from the (infinitesimally small) vicinities of these two
points. In order to quantify the leading contribution we fix $\de > 0$ small enough.
Since $u$ and $f$ are holomorphic we can deform the contour locally in a small
vicinity of $Q_\pm$ into contours $J_\pm^\de$ such that $w(\la) = 2 \p p T$ for
$\la \in J_\pm^\de$ and $v (\la) = \pm \de$ at the boundaries of $J_\pm^\de$.
Note that (for $\de$ small enough) $v$ is monotonic on $J_\pm^\de$, since it
has simple zeros at $Q_\pm$. It follows that
\begin{equation}
     \D I = T \int_{J_-^\de \cup J_+^\de} \rd \la \:
                    f(\la) \ln \Bigl(1 + \re^{- \frac{|v(\la)|}T} \Bigr)
            + \CO (T^\infty) \epp
\end{equation}
We parameterize $J_\pm^\de$ by $x = v(\la) \Leftrightarrow \la = v^{-1} (x)$.
Then
\begin{multline} \label{correctminus}
     T \int_{J_-^\de} \rd \la \: f(\la) \ln \Bigl(1 + \re^{- \frac{|v(\la)|}T} \Bigr)
        = T \int_\de^{-\de} \rd x \: \frac{f(v^{-1} (x))}{v'(v^{-1} (x))}
	                    \ln \Bigl(1 + \re^{- \frac{|x|}T} \Bigr) \\
        = - \frac{T^2 \p^2}{6} \frac{f(Q_-)}{u'(Q_-)} + \CO (T^4) \epp
\end{multline}
Here we have used that $v^{-1} (0) = Q_-$ and that $\ln \bigl(1 + \re^{- |x|}\bigr)$
is even. At $Q_+$ we can perform a similar calculation with the only difference
that $v(\la)$ is ascending in the direction of the contour, whence the sign will
be positive.
\end{proof}
This lemma can be directly applied to the function $\e(\la)$ defined by
(\ref{lieeps}) which meets the requirements of the lemma with ${\cal C}_\e
= {\cal C}_0 - \i \g/2$, $Q_\pm = \pm Q$ and $p = 0$. Thus,
\begin{equation} \label{epsinnlie}
     \e (\la) - \e_0 (\la) - T \int_{{\cal C}_0 - \i \g/2} \frac{\rd \m}{2\p \i} \:
                                  K(\la - \m) \ln \Bigl(1 + \re^{- \frac{\e(\m)}T} \Bigr)
				  = {\cal O} (T^2) \epp
\end{equation}
Comparing with (\ref{nlieu}), (\ref{nlieu0}) we see that asymptotically
for small $T$ these equations are indeed satisfied by $\e$. Equivalently,
$\fa_n (\la|\a) \sim \fa_0 (\la) \sim \re^{- \e(\la - \i \g/2)/T}$.
\begin{figure}[t]
\begin{center}
\includegraphics[width=.9\textwidth]{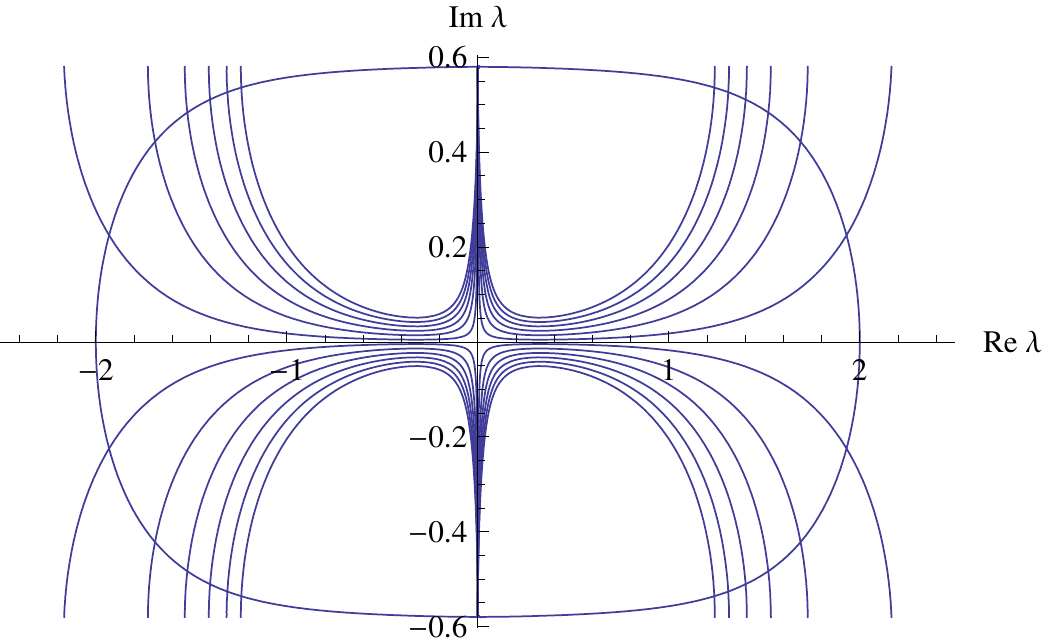}
\caption{\label{fig:br_anull} Solutions of the equation $\e (\la) =
- (2n -1 ) \p \i T$ in the complex plane depicted by the intersections
of the curve $\Re \e (\la) = 0$ (closed curve encircling the origin)
with the curves $\Im \e(\la) = - (2n - 1) \p T$ (open curves).
$J = 1$, $T = 0.01$, $\D = 0.4$, $h = 0.051$ and $n = -5, -4, \dots, 6$
in this example. The solutions form an infinite  sequence of points
on the line $\Re \e (\la) = 0$ with two limit points at $\pm \i \gamma/2$
(here $\gamma/2 = 0.580$). Shown are only the points farthest away from
these limit points.
}
\end{center}
\end{figure}

The function $u_0$ characterizing the dominant eigenvalue is approximated
by $\e$ with an error of second order in $T$. Hence, up to terms of the
order of $T^2$, the zeros of $1 + \fa_0$ are determined by
\begin{equation} \label{lowtfree}
     \e (\la) = - (2n - 1) \p \i T \epc \qd n \in {\mathbb Z} \epp
\end{equation}
Recall that those zeros that are located below the real axis are
the Bethe roots (shifted by $- \i \g/2$) of the dominant state of the
quantum transfer matrix in the Trotter limit (see e.g.\ \cite{GKS04a}).

The solutions of equation (\ref{lowtfree}) for a specific choice
of parameters are shown in Figure~\ref{fig:br_anull}. On the closed
curve encircling the origin and intersecting the real axis at
$\pm 2$ the function $\e$ is purely imaginary. This closed
curve $\G$ is intersected by the curves $\G_n$ of constant
imaginary part $\Im (\e (\la)) = - (2n - 1) \p T$. The intersection
points are the solutions of (\ref{lowtfree}). The function $\e$
has two zeros at $\pm Q = \pm 2$ and two poles at $\pm \i \g/2
= \pm 0.580 \i$ on $\G$. On each of the two arcs, running from
$- \i \g/2$ through $Q$ to $\i \g/2$ and back from $\i \g/2$ through
$- Q$ to $- \i \g/2$, the values of $\e$ increase from
$- \i \infty$ to $+ \i \infty$. Thus, there are infinitely many
solutions of (\ref{lowtfree}) on $\G$ clustering at $\pm \i \g/2$.
The figure shows only the solutions closest to the real axis.

From the picture we can also understand the qualitative behaviour
of the function $\fa_0$: $|\fa_0| = 1$ on $\G$, $|\fa_0 (\la)| > 1$ for
$\la$ inside $\G$ and $|\fa_0 (\la)| < 1$ for $\la$ outside $\G$.
Moreover, $\fa_0 $ is real negative on $\G_n$ with $\fa_0 (\la)
< - 1$ for $\la$ inside $\G$ and $-1 < \fa_0 (\la) < 0$ for $\la$
outside~$\G$. This means that the $\G_n$ are the canonical cuts
needed to construct the Riemann surface of the function $\ln (\fa_0)$.
It further follows that the function $1 + \fa_0$ is real and
negative with range between $- \infty$ and $0$ on those parts of
the contours $\G_n$ that are located inside $\G$. These are
therefore the cuts for the Riemann surface of $\ln(1 + \fa_0)$.

Equation (\ref{lowtfree}) resembles a momentum quantization condition
for free Fermions with momentum replaced by energy $- \e/T$. The
appearance of this terms seems rather natural in view of the
fact that `space and time direction' are interchanged in the
six-vertex model representing the partition function of the XXZ
chain within the quantum transfer matrix formalism.

As against $u_0$ the functions $u$ are generally approximated by $\e$
only in the strict limit $T \rightarrow 0$, when $(2n - 1) \p \i T$
becomes a continuous variable. In this limit the possible
particle and hole positions (shifted downward by $\i \g/2$) 
densely fill the curve $\Re \e (\la) = 0$. In order to obtain
discrete values to the order $T$, one has to take into account the
${\cal O} (T)$ contribution on the right hand side of equation
(\ref{nlieu}).

Assuming that for fixed particle and hole parameters the functions
$u$ admit the low-temperature asymptotic expansion
\begin{equation}
     u(\la) = \e (\la) + T u_1 (\la) + T^2 u_2 (\la) + \CO (T^3)
\end{equation}
and that there are $p \in {\mathbb Z}$ and $Q_\pm = \pm Q + T Q^{(1)}_\pm
+ \CO (T^2)$ such that $u(Q_\pm) = 2 \p \i p T$ we conclude that 
\begin{equation} \label{boundariesot}
     Q_\pm = \pm Q \mp \frac{\uq_1 (\pm Q)}{\e' (Q)} T + \CO (T^2) \epc
\end{equation}
where $\uq_1 (\la) = u_1 (\la) - 2 \p \i p$ by definition. Since
$u$ is close to $\e$ as long as $T$ is small enough we may apply
Lemma~\ref{lem:sommerfeld} to $u$. Using a contour ${\cal C}_u$
which is a deformation of ${\cal C}_0 -\i \g/2$ such that it passes
through $Q_\pm$ and introducing the notation $\uq = u - 2 \p \i p T$
the lemma implies that
\begin{equation} \label{ubart2}
     \uq (\la) = \e_0 (\la) + T r_1 (\la) + T^2 r_2 (\la)
                 + \int_{-Q}^{Q} \frac{\rd \m}{2 \p \i} K(\la - \m) \uq (\m)
		 + \CO (T^3) \epc
\end{equation}
where
\begin{align} \label{defr12}
     & r_1 (\la) = 2 \p \i (\a' - s/2 - p) + \sum_{j=1}^{n_p} \th (\la - x_j^p)
                   - \sum_{j=1}^{n_h} \th (\la - x_j^h) \epc \\[1ex]
     & r_2 (\la) = \frac{\i \p}{4 \e' (Q)}
                   \biggl[ K(\la - Q) \biggl(\frac13 + \frac{\uq_1^2 (Q)}{\p^2}\biggr)
		         + K(\la + Q) \biggl(\frac13 + \frac{\uq_1^2 (-Q)}{\p^2}\biggr)
		         \biggr]
\end{align}
with $x_j^{p, h} = \la_j^{p, h} - \i \g/2$.

Neglecting the ${\cal O} (T^3)$ terms in (\ref{ubart2}) we obtain a linear
integral equation for $\uq$. Due to its linearity we can express its
solution in terms of standard functions known from the description of the
ground state properties of the XXZ chain. We obtain
\begin{subequations}
\label{u1u2exp}
\begin{align} \label{u1exp}
     \uq_1 (\la) & = 2 \p \i \biggl[ (\a' -s/2 - p) Z(\la)
                        + \sum_{j=1}^{n_h} \Ph (\la, x_j^h)
                        - \sum_{j=1}^{n_p} \Ph (\la, x_j^p) \biggr] \epc \\[1ex]
			\label{u2exp}
     u_2 (\la) & = \frac{\i \p}{4 \e' (Q)}
                   \biggl[ R(\la,Q) \biggl(\frac13 + \frac{\uq_1^2 (Q)}{\p^2}\biggr)
		         + R(\la,-Q) \biggl(\frac13 + \frac{\uq_1^2 (-Q)}{\p^2}\biggr)
		         \biggr] \epp
\end{align}
\end{subequations}
The functions appearing here are the dressed charge function $Z$, the
dressed phase $\Ph$ and the resolvent $R$, satisfying the linear integral
equations
\begin{subequations}
\label{zphr}
\begin{align} \label{liz}
     Z (\la) & = 1 + \int_{-Q}^{Q} \frac{\rd \m}{2 \p \i} K(\la - \m) Z (\m) \epc \\
     \Ph (\la, \n) & = - \frac{\th (\la - \n)}{2 \p \i}
                       + \int_{-Q}^{Q} \frac{\rd \m}{2 \p \i} K(\la - \m) \Ph (\m, \n)
		         \epc \label{liph} \\
     R (\la, \n) & = K(\la - \n)
                     + \int_{-Q}^{Q} \frac{\rd \m}{2 \p \i} K(\la - \m) R (\m, \n) \epp
\end{align}
\end{subequations}
For later convenience we also introduce the root density $\r$ as the
solution of
\begin{equation} \label{rho}
     \r (\la) = - \frac{\re (\la + \i \g/2)}{2 \p \i}
                + \int_{-Q}^{Q} \frac{\rd \m}{2 \p \i} K(\la - \m) \r (\m) \epp
\end{equation}

Knowing $u_1$ we know the subsidiary conditions (\ref{subs}) for
$\fa_n (\cdot|\a)$ to linear order in $T$,
\begin{equation} \label{baelint}
     \e(x_j^{p,h}) + T u_1 (x_j^{p,h}) = - (2 n_j^{p, h} - 1) \i \p T \epp
\end{equation}
We insert the explicit expression (\ref{u1exp}) for $\uq_1$ into this
equation and obtain a set of coupled nonlinear algebraic equations
for the particle and hole parameters,
\begin{multline} \label{dressedbaes}
     \frac{\e(x_j^{p,h})}{2 \p \i T} =
                        - n_j^{p,h} + 1/2 - p \\
			- (\a' - s/2 - p) Z(x_j^{p,h})
                        + \sum_{k=1}^{n_p} \Ph (x_j^{p,h}, x_k^p)
                        - \sum_{k=1}^{n_h} \Ph (x_j^{p,h}, x_k^h) \epp
\end{multline}
These equations may be interpreted as a dressed version of the logarithmic
form of the Bethe ansatz equations (with dressed momentum replaced by
dressed energy and $1/L$ replaced by~$T$). The bare two-particle scattering
phases are replaced by the dressed phases and the dressed charge
appears in addition. Equations (\ref{dressedbaes}) have to be solved
numerically for the particle and hole parameters $x_j^p$ and $x_j^h$.

Simplifications occur in two cases. In the XX or free fermion case,
$\g = \p/2$, $\D = 0$, the dressed phases vanish, the dressed charge
equals one, and the dressed energy turns into the bare energy $\e_0$.
Thus,
\begin{equation} \label{freefermi}
     \e_0 (x_j^{p, h})
        = - (2 n_j^{p,h} - 1 - s + \a') \p \i T \epp
\end{equation}
The same set of decoupled equations is obtained from the non-linear
integral equations (\ref{nlieu}) if one sets $\g = \p/2$, meaning
that in the XX case it is valid for any $T$. In this case all solutions
fall into two classes depending on whether $s/2$ is integer or half-odd
integer. Comparing (\ref{freefermi}) and (\ref{lowtfree}) we observe that
the description of the dominant state of the quantum transfer matrix
in terms of the zeros of the auxiliary function $1 + \fa_0$ is very
close to the free fermion paradigm. This provides a useful `almost
free fermion picture' for the understanding of the excited states
of the quantum transfer matrix at low temperatures.

Further simplifications for generic $\g \in (0, \p/2)$ occur if
we restrict ourselves to excitations close to the Fermi surface
consisting of the two points $\pm Q$. Our low-temperature analysis
of the correlation functions is based on the hypothesis that these
excitations contribute predominantly to the large-distance asymptotics.
More precisely, we shall restrict ourselves in the following
to particle and hole parameters which collapse to the Fermi points
$\pm Q$ as $T$ goes to zero,
\begin{equation} \label{abovefermi}
     x_j^{p, h} = \pm Q + {\cal O} (T) \epp
\end{equation}
We denote their numbers by $n_p^\pm$ and $n_h^\pm$, respectively, and
define the particle-hole disbalance at the left Fermi point $-Q$ by
\begin{equation} \label{npmnhmell}
     \ell = n_h^- - n_p^- \epp
\end{equation}
Inserting the lowest order approximation $x_j^{p, h} = \pm Q$ into
(\ref{u1exp}) we obtain the leading low-temperature approximation to
$\uq_1 (\la)$ which we denote $\uq_1^{(\ell)} (\la)$. We shall write
it as
\begin{equation} \label{u1ell}
     \uq_1^{(\ell)} (\la) = 2 \p \i \bigl( w(\la) + \a' - p - s/2 \bigr) \epc
\end{equation}
where
\begin{equation}
     w(\la) = (\a' - p - \ell) \bigl(Z(\la) - 1 \bigr)
	+ \frac{s}{2} \bigl(\Ph (\la, Q) + \Ph (\la, -Q)\bigr) \epp
\end{equation}
The function $\uq_1^{(\ell)}$ determines the points $Q_\pm$ in
(\ref{boundariesot}) to linear order in $T$,
\begin{equation} \label{boundariesotdet}
     Q_\pm = \pm Q \mp \frac{\uq_1^{(\ell)} (\pm Q)}{\e' (Q)} T + \CO (T^2) \epp
\end{equation}

Using (\ref{u1ell}) in (\ref{baelint}), the equations for the particle
and hole parameters decouple. We obtain
\begin{equation} \label{brnearq}
     \frac{\e(x_j^{p,h})}{2 \p \i T} =
        - n_j^{p,h} + \frac{1 + s}2 - \a' - w(x_j^{p,h}) \epp
\end{equation}
Comparison with (\ref{freefermi}) shows that the function $w$ determines
the deviation from free-Fermion like behaviour to linear order in $T$.
\begin{figure}[t]
\begin{tabular}{@{}cc@{}}
\includegraphics[width=.48\textwidth]{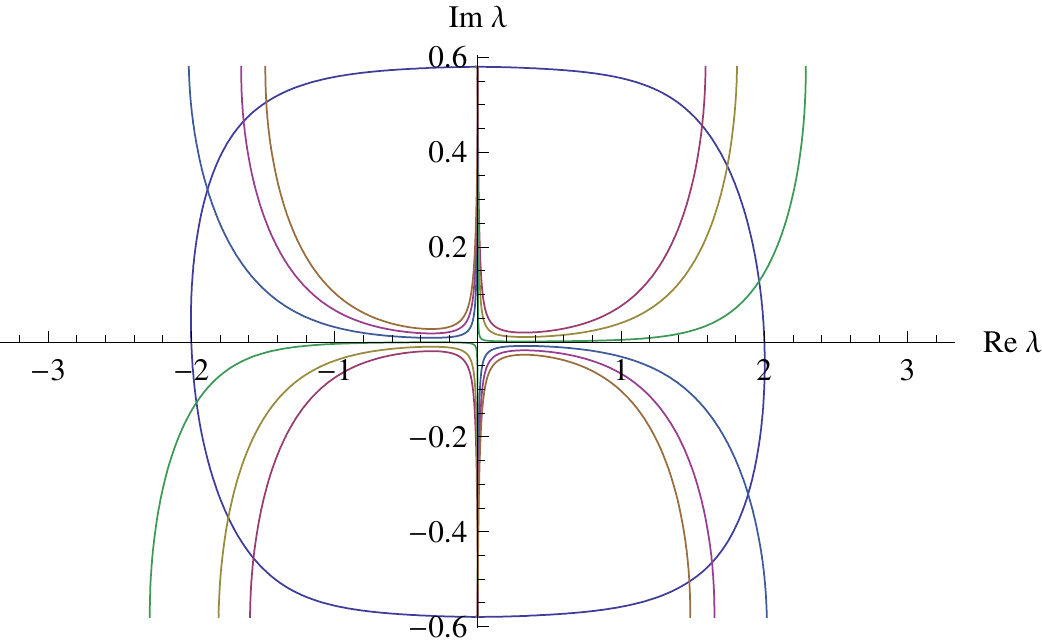} &
\includegraphics[width=.48\textwidth]{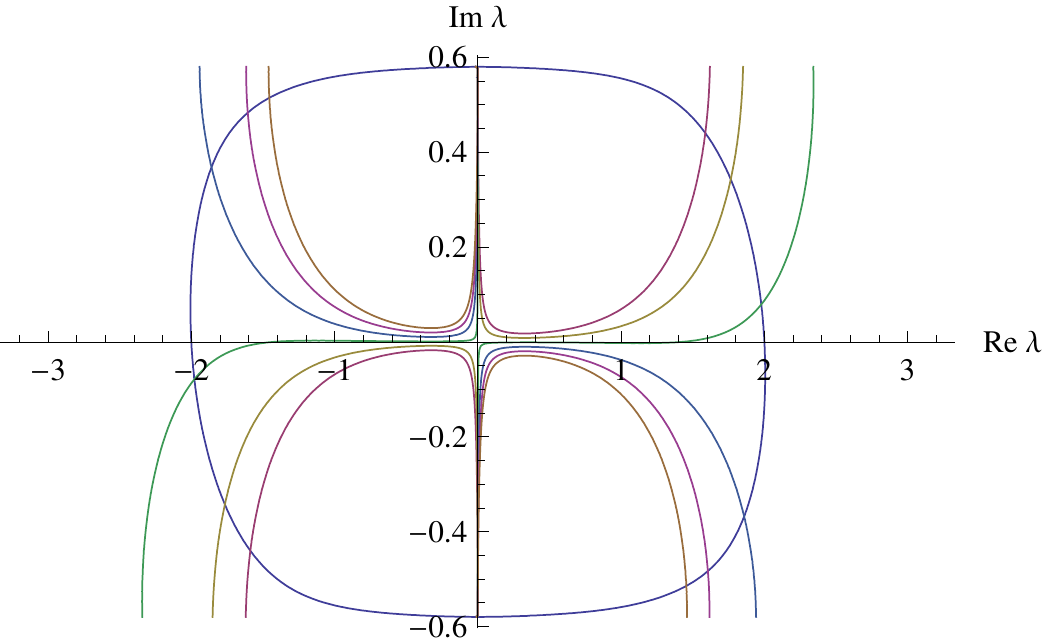} \\
$p + \ell = 2$, $s = 0$ &
$p + \ell = 3$, $s = 0$ \\[3ex]
\includegraphics[width=.48\textwidth]{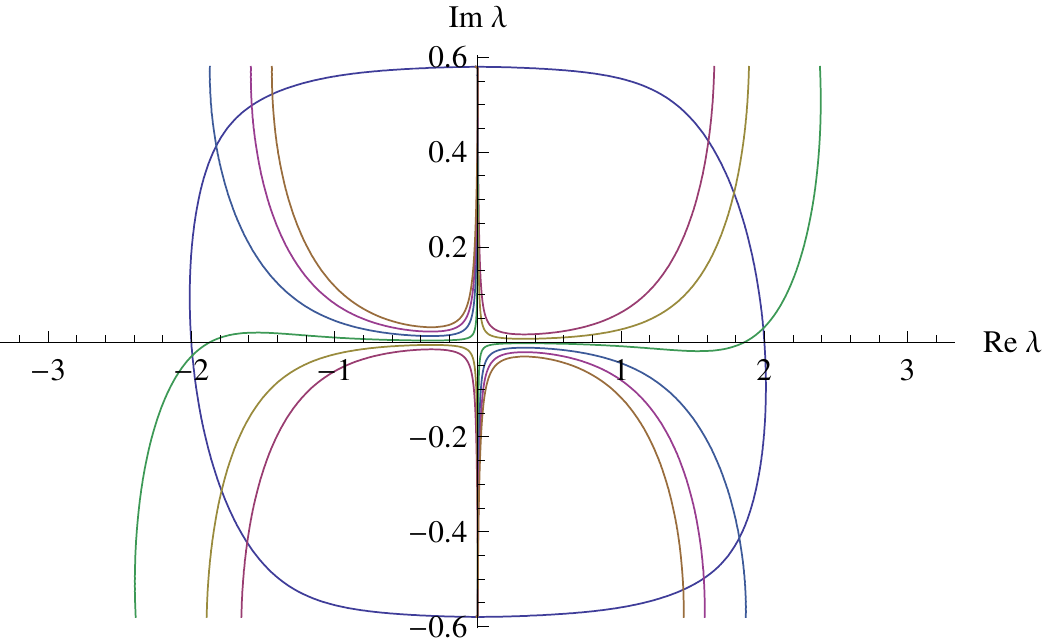} &
\includegraphics[width=.48\textwidth]{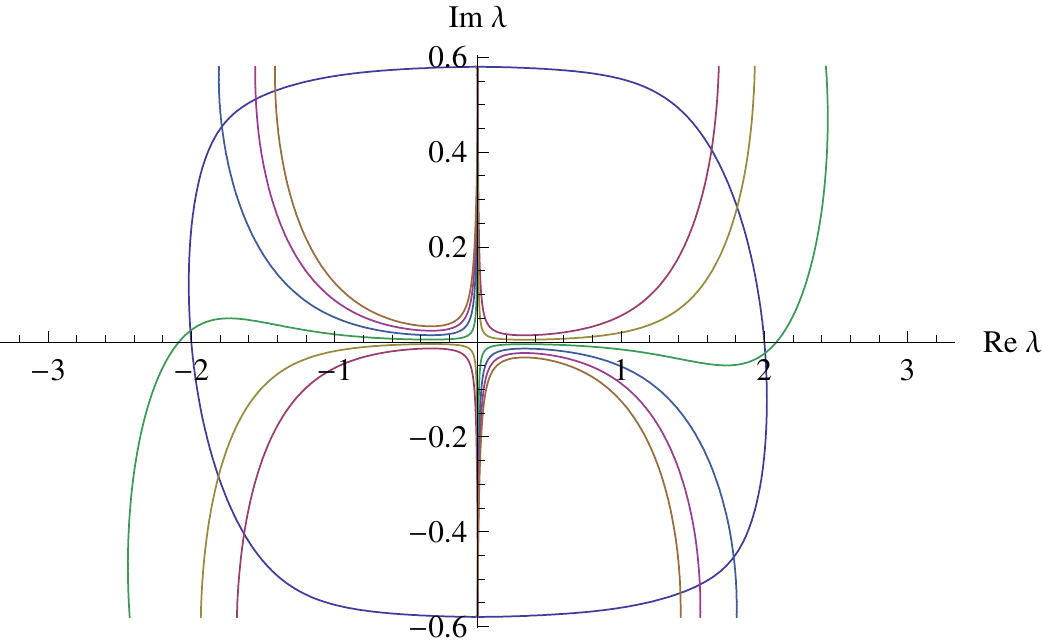} \\
$p + \ell = 4$, $s = 0$ &
$p + \ell = 5$, $s = 0$
\end{tabular}
\caption{\label{fig:br_an} Graphical solutions of equation (\ref{brnearq})
in the complex plane for various values of $p + \ell$ and $s = 0$. Parameter
values are the same as in Figure~\ref{fig:br_anull}: $J = 1$, $T = 0.01$,
$\D = 0.4$, $h = 0.051$. Shown are the solutions closest to the real axis.
}%
\end{figure}%
As we further see from (\ref{brnearq}), to this order, the
`excitations above the Fermi surface' (\ref{abovefermi}) fall
into classes parameterized by two sets of integers, $p + \ell$
and $s$. In our case $s = 0$ or $s = 1$ for the longitudinal
and transversal two-point functions. For fixed $s$ all excitations
above the Fermi surface are obtained by letting $p + \ell$ run
through all integers and by calculating the associated
particle-hole patterns satisfying (\ref{npnhs}) and (\ref{npmnhmell})
from (\ref{brnearq}). We shall see below that the amplitudes that
remain after summing over all particle-hole patterns for fixed
$\ell$ depend on $p$ and $\ell$ indeed only through their sum
$p + \ell$.
 
So far we did not discuss the meaning of $p$. It entered our
calculation when we used Lemma~\ref{lem:sommerfeld} to determine
the low-temperature approximation to $\ln \fa_n (\cdot|\a)$.
As we see from the lemma $- 2 \p p$ is the phase of $\fa_n (Q_\pm|\a)$.
Thus, it is tightly connected with the choice of the reference
contour. In principle, we would like to choose the real
axis (corresponding to the upper part of the canonical
contour) as a reference contour. Then $p$ would be determined
by the condition that both, $Q_-$ and $Q_+$ are located
at the zeros of $1 - \fa_n (\cdot|\a)$ that are closest to the
real axis (\textit{i.e.} to $\pm Q$). We shall see below that
such a choice of reference contour is possible for $s = 0$.
In the general case the reference contour can pass only
either through $- Q$ or through $Q$. Below we shall choose
$- Q$. Then $Q_\pm$ are determined by (\ref{boundariesotdet}).

In order to get an intuitive understanding of $p$ let us
consider an example. In Figure~\ref{fig:br_an} we have depicted
the solutions of (\ref{brnearq}) by plotting the real
and imaginary parts of the difference between left and right
hand side of the equation. For the same set of parameters as
in Figure~\ref{fig:br_anull} we have set $s = 0$ and have
increased the value of $p + \ell$ in unit steps from $2$ to $5$.
For $p + \ell = 0$ we would obtain again Figure~\ref{fig:br_anull}.
Possible positions of particle and hole parameters $x_j^p$, $x_j^h$
are the intersection points of the open contours with the closed
contours that are located above and below the real line, respectively.
In this case the real axis (corresponding to the upper part of
the contour ${\cal C}_0 - \i \g/2$ and run through in negative
direction) is a possible reference contour, since left and right
intersection points of the real axis with the curve $\Re u = 0$
are located on the same sheet of the Riemann surface of $\ln
\fa_n (\cdot |\a)$, whose branch cuts are the open contours in
the figure. As soon as $p + \ell$ is as large as $5$ in our
example, our reference contour intersects one of the branch
cuts of $\ln \fa_n (\cdot|\a)$ before it intersects the closed
contour corresponding to the real part of equation (\ref{brnearq}).
Then the phase of $\fa_n (Q_\pm| \a)$ in units of $- 2 \p$, which
is the value of $p$, is incremented by one to $p = 1$. Hence,
in the first three pictures we have $p = 0$ and $\ell = 2, 3, 4$,
but in the last one $p = 1$ and $\ell = 4$.
\begin{figure}[t]
\begin{tabular}{@{}cc@{}}
\includegraphics[width=.48\textwidth]{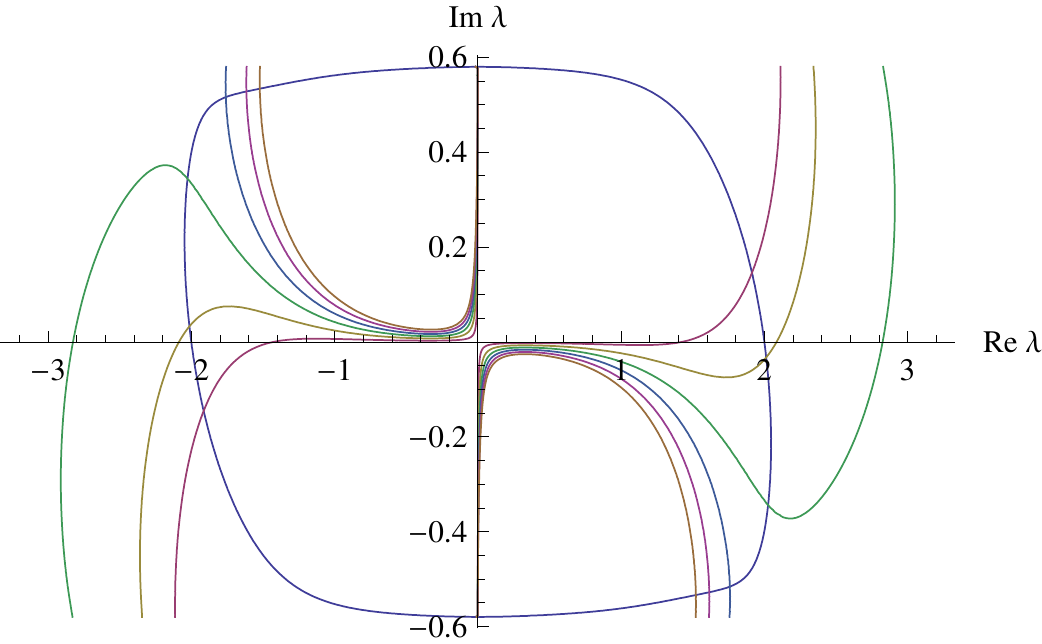} &
\includegraphics[width=.48\textwidth]{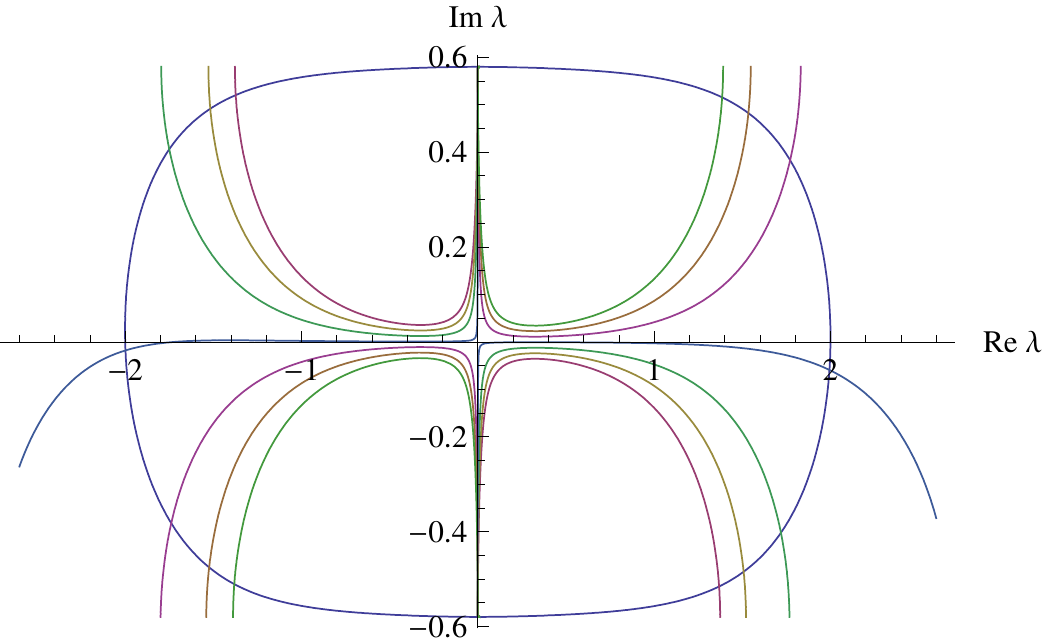} \\[.5ex]
$p + \ell = 15$, $s = 0$ & $p + \ell = 0$, $s = 1$
\end{tabular}
\caption{\label{fig:br_an2} Graphical solutions of equation (\ref{brnearq})
in the complex plane for various values of $p + \ell$ and $s$. Parameter
values in the right panel $J = 1$, $T = 0.0125$, $\D = 0.4$, $h = 0.051$
and $\a = 0.3$. In the left panel the temperature is reduced to $T = 0.005$
and $\a = 0$.
}%
\end{figure}%
The left panel in Figure~\ref{fig:br_an2} shows an example where $p = 2$.

For non-zero $s$ the pictures become more asymmetric. An example is
shown in the right panel of Figure~\ref{fig:br_an2}, where $p + \ell = 0$
and $s = 1$. A contour along the real line, entering the picture
from the right and leaving to the left, now crosses one of the lines
of constant imaginary part. Let $Q_\pm$ be the zeros of $1 - \fa_n (\cdot| \a)$
closest to the real axis. Then it follows, due to the crossing of the
line, that $\ln \fa_n (Q_+|\a) - \ln \fa_n (Q_-|\a) = - 2 \p \i$. Hence,
the contour does not fit the requirements of Lemma~\ref{lem:sommerfeld}.
A contour ${\cal C}_{0,1}$ as close as possible to the real axis and
meeting the requirements of the lemma must pass one of the two intersection
points with negative imaginary part closest to the real axis from 
above and the other one from below.

From equation (\ref{brnearq}) we obtain the deviation of the
particle and hole parameters $x_j^{p, h}$ from the Fermi points
to linear order in $T$,
\begin{equation} \label{phlint}
     x_j^{p, h} = \pm Q
        - \Bigl(n_j^{p,h} - \frac{1 + s}2 + \a' + w(\pm Q) \Bigr)
	  \frac{2 \p \i T}{\e' (Q)} \epp
\end{equation}
In this approximation the particle and hole parameters are located
on lines perpendicular to the real axis and intersecting it at $\pm Q$.

We now distinguish particle and hole parameters pertaining the the
right and left Fermi edges $\pm Q$, writing $x_j^{p \pm}$ and $x_j^{h \pm}$,
respectively. The corresponding quantum numbers are denoted $n_j^{p \pm}$
and $n_j^{h \pm}$. We reparameterize these integers by positive
integers $p_j^\pm$, $h_j^\pm$, setting
\begin{equation} \label{phqnumber}
     n_j^{p \pm} = - p_j^\pm - p + 1 \epc \qd n_j^{h \pm} = h_j^\pm - p \epp
\end{equation}
In order to obtain a consistent interpretation of the positive numbers
$p_j^\pm$ and $h_j^\pm$ as `particle and hole quantum numbers' we have
to fix $p$ in such a way that e.g.\ $\Im x_j^{p-} > 0 > \Im x_j^{h-}$
for $p_j^- = h_j^- = 1$. Using (\ref{u1ell}) and (\ref{phlint}), this
is equivalent to demanding that
\begin{equation} \label{choosep}
     \bigl| \Im \uq_1^{(\ell)} (-Q) \bigr| < \p \epp
\end{equation}
Inserting now (\ref{boundariesotdet}) and (\ref{phqnumber}) into (\ref{phlint})
the particle and hole parameters become parameterized as
\begin{subequations}
\label{hppms}
\begin{align}
     & x_j^{p \pm} = Q_\pm + \biggl(p_j^\pm - \2 \biggr) \frac{2 \p \i T}{\e' (Q)}
                   + {\cal O}(T^2) \epc \qd \\[1ex]
     & x_j^{h \pm} = Q_\pm - \biggl(h_j^\pm - \2 \biggr) \frac{2 \p \i T}{\e' (Q)}
                   + {\cal O}(T^2)
\end{align}
\end{subequations}
where $p_j^\pm, h_j^\pm \in {\mathbb Z}_+$. The inequality (\ref{choosep})
guarantees that $Q_-$ is as close to the real axis as possible. If $s = 0$
we have $\uq_1^{(\ell)} (Q) = \uq_1^{(\ell)} (-Q)$ and (\ref{choosep})
automatically holds at the right Fermi edge as well. This is no longer true
for $s \ne 0$. In that case, using the well known identities
\cite{KoSl98}
\begin{equation}
     {\cal Z} = 1 + \Ph (Q, Q) - \Ph (Q, -Q) \epc \qd
     \frac{1}{\cal Z} = 1 + \Ph (Q, Q) + \Ph (Q, -Q) \epc
\end{equation}
we can only conclude that
\begin{equation}
     0 < \Im \uq_1^{(\ell)} (Q) - \Im \uq_1^{(\ell)} (-Q)
       = 2 s \p \bigl(1/{\cal Z} - 1 \bigr) < s \p \epp
\end{equation}
The inequalities hold, since $1/\sqrt 2 < {\cal Z} < 1$ as long as
$\g \in (0, \p/2)$ (see e.g.\ \cite{DGK13bpp}). Thus, for $s = 1$
it may happen that $\p < \Im \uq_1^{(\ell)} (Q) < 2 \p$. For
$p_j^+ = 1$ the latter implies that $\Im x_j^{p +} =
\bigl(\p - \Im \uq_1^{(\ell)} (Q) \bigr) T /\e'(Q) < 0$,
\textit{i.e.}\ according to our definition the lowest particle
excitation at the right Fermi edge may correspond to an $x_j^{p+}$
below the real axis. Of course, this is in accordance with the
example considered above, illustrated in the right panel of
Figure~\ref{fig:br_an2}, and with our general picture which
includes that we have to choose the contour ${\cal C}_{0, 1}$
carefully.

Before closing this section let us stress again that $p$ is
an auxiliary parameter associated with the choice of the
integration contour, in which we have a certain freedom and
which does not influence the final result of our calculation.
The true parameter entering the classification of the elementary
excitations is the sum $p + \ell$, since only this parameter enters
the definition of the function $w$ which determines the particle
and hole parameters through (\ref{brnearq}).

For the low-temperature analysis of the eigenvalue ratios in the next
subsection we shall need $\uq_1$ up to the first order in $T$.
We insert (\ref{u1ell}), (\ref{hppms}) into (\ref{u1exp}) and use
$\6_\n \Ph (\la, \n) = R(\la, \n)/2 \p \i$ to obtain
\begin{multline} \label{u1expexp}
     \uq_1 (\la) = \uq_1^{(\ell)} (\la) + \frac{2 \p \i T}{\e'(Q)}
                   \biggl\{ R(\la,Q) \biggl[ \frac{(\ell - s) \uq_1^{(\ell)} (Q)}{2 \p \i}
		      - \sum_{j=1}^{n_h^+} \Bigl(h_j^+ - \frac12\Bigr)
		      - \sum_{j=1}^{n_p^+} \Bigl(p_j^+ - \frac12\Bigr) \biggr] \\
                    + R(\la, -Q) \biggl[ \frac{\ell \uq_1^{(\ell)} (-Q)}{2 \p \i}
		      - \sum_{j=1}^{n_h^-} \Bigl(h_j^- - \frac12\Bigr)
		      - \sum_{j=1}^{n_p^-} \Bigl(p_j^- - \frac12\Bigr) \biggr] \biggr\}
		      + \CO (T^2) \epp
\end{multline}
Equations (\ref{u2exp}) and (\ref{u1expexp}) together with the corresponding
linear integral equations then determine $\ln(\fa_n (\la|\a))$ up to the order $T$.
\subsection{Correlation lengths and universal part of the amplitudes}
Using the results of the previous subsection it is not difficult to
calculate the leading order low-temperature contribution to the
eigenvalue ratios for the transversal correlation functions ($s = 1$). Setting
$\a'' = \a - p - \ell$ we obtain
\begin{multline}
     \r_n (0|\a) =
        q^\a \exp \biggl\{ \i \p - 2 \i \a'' k_F \\
	- \frac{2 \p T}{v_0} \biggl[ \a''^2 {\cal Z}^2
	+ \frac{1}{4 {\cal Z}^2} - \ell^2 + \ell -1
	+ \sum_{j=1}^{n_h} h_j + \sum_{j=1}^{n_p}  (p_j - 1) \biggr]
	  \biggr\} + \CO(T^2) \epp
\end{multline}
Here we have introduced the Fermi momentum $k_F$ and the sound
velocity $v_0$ as
\begin{equation} \label{deffermimom}
     k_F = 2 \p \int_0^Q \rd \la \: \r(\la) \epc \qd
     v_0 = \frac{\e'(Q)}{2 \p \r(Q)} \epp
\end{equation}

Starting from the expression (\ref{amps}) for the amplitudes and using
the insight from the previous subsection we can also calculate the
leading low-temperature asymptotics of the universal part of the amplitudes.
After slightly tedious calculations we obtain
\begin{equation}
     A_n^{(0)} (\a) = - \i q^\a \re^{E(-2Q)} A_- (\a) A_+ (\a) \epc
\end{equation}
where
\begin{multline} \label{lowtamps}
     A_\pm (\a) =
        \exp \biggl\{ C_\pm [w] + \4 \th (2 Q) \Bigl[ 2 \a''^2 {\cal Z}^2
                       + \frac1{2 {\cal Z}^2} + 1
		       - \a'' \Bigl( {\cal Z} - \frac1{\cal Z} \Bigr)\Bigr]\biggr\}
		       \\[1ex] \times
	\biggl( \frac{2 \p T \re^{- E (2Q)}}
	             {\e' (Q) \sh(\h)} \biggr)^{\a''^2 {\cal Z}^2 + \frac1{4 {\cal Z}^2}}
	\frac{G \bigr(3/2 \pm (w(\pm Q) + \a'') \bigr)}
	     {G \bigr(1/2 \mp (w(\pm Q) + \a'') \bigr)} \\[1ex] \times
        G^2 \biggr(1 \mp \frac{\uq_1^{(\ell)} (\pm Q)}{2 \p \i} \biggr)
        \biggl( \frac1{\p}
	        \sin \biggl( \frac{\uq_1^{(\ell)} (\pm Q)}{2 \i} \biggr) \biggr)^{2 n_h^\pm}
        {\cal R}_{n_h^\pm, n_p^\pm}
	   \biggl( \{h_j^\pm\}, \{p_j^\pm\} \Big|
	           \pm \frac{\uq_1^{(\ell)} (\pm Q)}{2 \p \i} \biggr).
\end{multline}
In this expression $G$ is the Barnes $G$-function and $C_\pm$ are the functionals
\begin{multline}
     C_\pm [v] = \4 \int_{- Q}^Q \rd \la \int_{- Q}^Q \rd \m \:
		 \bigl( v' (\la) v (\m) - v (\la) v' (\m) \bigr) \re (\la - \m) \\
                 \pm \bigl(v (\pm Q) \pm 1 + \a''\bigr) \int_{- Q}^Q \rd \la
	         \bigl(v (\la) -  v (\pm Q)\bigr) \re (\la \mp Q) \epp
\end{multline}
The functions ${\cal R}$ comprise the dependence on the particle-hole
quantum numbers,
\begin{multline} \label{rcomb}
     {\cal R}_{n_1, n_2} \bigl( \{h_j\} , \{p_j\} \big| v \bigr)
        = \frac{\prod_{1 \le j < k \le n_1}
	        (h_j - h_k)^2 \prod_{1 \le j < k \le n_2} (p_j - p_k)^2}
	       {\prod_{j=1}^{n_1} \prod_{k=1}^{n_2} (h_j + p_k - 1)^2} \\ \times
          \Biggl[ \prod_{j=1}^{n_1}
	       \frac{\G^2 ( h_j + v)} {\G^2 (h_j)} \Biggr]
          \Biggl[ \prod_{j=1}^{n_2}
	       \frac{\G^2 ( p_j - v)} {\G^2 (p_j)} \Biggr] \epp
\end{multline}

The hardest part of the calculation leading to (\ref{lowtamps}) is the
evaluation of the singular integrals in the exponent on the right
hand side of (\ref{amps}). It can be achieved by means of the following
lemmas.
\begin{lemma} \label{lem:sing1}
Let $u$ and ${\cal C}_u$ be subject to the same assumptions as in
Lemma~\ref{lem:sommerfeld}. Let $\la_+$ be located above ${\cal C}_u$
and $\la_-$ below. Then the Cauchy-type integral
\begin{equation}
     I_u (\la_\pm) = \int_{{\cal C}_u} \rd \la \: \cth(\la - \la_\pm)
                        \ln \Bigl(1 + \re^{- \frac{u(\la)}T} \Bigr)
\end{equation}
admits a low-temperature expansion whose form depends on the distance
of $\la_\pm$ from the zeros $Q_\pm$ of the real part of $u$.

If $\la_\pm$ are uniformly away from $Q_\pm$, then
Lemma~\ref{lem:sommerfeld} applies, and
\begin{equation}
     I_u (\la_\pm) =
        - \int_{Q_-}^{Q_+} \rd \la \: \cth(\la - \la_\pm) \frac{\uq (\la)}{T}
	+ \CO (T) \epp
\end{equation}

For $\de > 0$ define $V_\pm = \bigl\{ z \in {\mathbb C} \big| |u(z)| < \de/2,
\text{$z$ close to $Q_\pm$} \bigr\}$. If $\la_\pm \in V_+$, then there is
a $\de > 0$ and independent of $\la_\pm$ such that
\begin{multline} \label{iuright}
    I_u (\la_\pm) =  - \int_{Q_-}^{Q_+} \rd \la \:
                       \cth(\la - \la_\pm) \frac{\uq (\la) - \uq (\la_\pm)}{T}
                     \ \mp 2 \p \i \ln \biggl\{ \G \biggl(
		        \2 \pm \frac{\uq (\la_\pm)}{2 \p \i T} \biggr) \biggr\} \\
      \pm \p \i \ln (2 \p) + \frac{\uq (\la_\pm)}T \biggl\{
                     \ln \biggl(\frac{\uq (\la_\pm)}{\pm 2 \p \i T} \biggr) - 1
                     - \ln \biggl(\frac{\sh (Q_+ - \la_\pm)}{\sh (Q_- - \la_\pm)} \biggr)
		     \biggr\} + \CO (T) \epp
\end{multline}
If $\la_\pm \in V_-$, then
\begin{multline} \label{iuleft}
    I_u (\la_\pm) =  - \int_{Q_-}^{Q_+} \rd \la \:
                       \cth(\la - \la_\pm) \frac{\uq (\la) - \uq (\la_\pm)}{T}
                     \ \mp 2 \p \i \ln \biggl\{ \G \biggl(
		        \2 \mp \frac{\uq (\la_\pm)}{2 \p \i T} \biggr) \biggr\} \\
      \pm \p \i \ln (2 \p) - \frac{\uq (\la_\pm)}T \biggl\{
                     \ln \biggl(\frac{\uq (\la_\pm)}{\mp 2 \p \i T} \biggr) - 1
                     + \ln \biggl(\frac{\sh (Q_+ - \la_\pm)}{\sh (Q_- - \la_\pm)} \biggr)
		     \biggr\} + \CO (T) \epp
\end{multline}
\end{lemma}
\begin{lemma} \label{lem:sing2}
\begin{multline}
     - \int_{{\cal C}_{0,s}} \rd \la \: \int_{{\cal C}_{0,s}'} \rd \m \:
          z(\la) \cth' (\la - \m) z(\m) = \\
          + \ln \biggl\{
	           G \biggr(1 + \frac{\uq_1^{(\ell)} (Q)}{2 \p \i} \biggr)
	           G \biggr(1 - \frac{\uq_1^{(\ell)} (Q)}{2 \p \i} \biggr)
	           G \biggr(1 + \frac{\uq_1^{(\ell)} (- Q)}{2 \p \i} \biggr)
	           G \biggr(1 - \frac{\uq_1^{(\ell)} (- Q)}{2 \p \i} \biggr)
	        \biggr\} \\
	  + C_1 \biggl[ \frac{\uq_1^{(\ell)}}{2 \p \i} \biggr]
	  - \biggl( \biggl( \frac{\uq_1^{(\ell)} (Q)}{2 \p \i} \biggr)^2
	     + \biggl( \frac{\uq_1^{(\ell)} (-Q)}{2 \p \i} \biggr)^2 \biggr)
	     \ln \biggl( \frac{\e' (Q) \sh(2 Q)}{2 \p T} \biggr)
		+ o (1) \epc
\end{multline}
where $o(1)$ denotes terms that go to zero as $T \rightarrow 0^+$, and
\begin{multline}
     C_1 [v] = \2 \int_{- Q}^Q \rd \la \int_{- Q}^Q \rd \m \:
		  \frac{v' (\la) v (\m) - v (\la) v' (\m)}{\tgh (\la - \m)} \\
               + v (Q) \int_{- Q}^Q \rd \la
	         \frac{v (\la) -  v (Q)}{\tgh (\la - Q)}
               - v (-Q) \int_{- Q}^Q \rd \la
	         \frac{v (\la) -  v (-Q)}{\tgh (\la + Q)} \epp
\end{multline}
\end{lemma}
Proofs of Lemma~\ref{lem:sing1} and Lemma~\ref{lem:sing2} are provided
in Appendix~\ref{app:proofs}.

\subsection{Determinant part and factorized part of the amplitudes}
For the calculation of the determinant part and of the `factorized
part', by which we mean the product of functions $\overline{G}_\pm$
in (\ref{amp}), we can closely follow Appendix D of \cite{DGK13a}.
An important property of these contributions to the amplitudes is
that they do not depend on the particle and hole quantum numbers.
They are functionals of $w$ and depend only on $p + \ell$.
Anticipating this fact we shall write
\begin{equation} \label{defdl}
     {\cal D} (p + \ell) = \lim_{T \rightarrow 0}
        \frac{\det_{\rd m^\a_+, {\cal C}_n} \bigl\{ 1 - \widehat{K}_{1-\a} \bigr\}
	      \det_{\rd m^\a_-, {\cal C}_n} \bigl\{ 1 - \widehat{K}_{1+\a} \bigr\}}
	     {\det_{\rd m^\a_0, {\cal C}_n} \bigl\{ 1 - \widehat{K} \bigr\}
	      \det_{\rd m, {\cal C}_n} \bigl\{ 1 - \widehat{K} \bigr\}} \epp
\end{equation}

The measures of the determinants in the denominator do not contain
$\r_n (\cdot |\a)$. Their zero temperature limit is readily understood
by recalling that the weight functions $1/\bigl( 1 + \fa_0^{-1} \bigr)$
and $1/\bigl(1 + \fa_n^{-1} (\cdot|\a) \bigr)$ turn into the characteristic
functions of the sub-contour $\i \g/2 + [-Q, Q]$ on ${\cal C}_0$. Thus,
\begin{equation}
     \lim_{T \rightarrow 0}
        \det_{\rd m^\a_0, {\cal C}_n} \bigl\{ 1 - \widehat{K} \bigr\} = 
     \lim_{T \rightarrow 0}
	\det_{\rd m, {\cal C}_n} \bigl\{ 1 - \widehat{K} \bigr\} = 
	\det_{\rd \la/2 \p \i, [-Q, Q]} \bigl\{ 1 - \widehat{K} \bigr\}
\end{equation}
(for more details see \cite{DGK13a}). Here the right hand side does not
depend on any characteristic of the excitation and is a simple function
of the magnetic field.

With the determinants in the numerator we proceed as in Appendix D of
\cite{DGK13a}. Using an idea borrowed from \cite{BoGo10} we decomposed
the measures into
\begin{multline} \label{repweights}
     \rd m_\pm^\a (\la) =
        \frac{\rd \la}{2 \p \i} \Biggl[
	\biggl( \r_n (\la|\a) \frac{1 + \fa_0 (\la)}{1 + \fa_n (\la|\a)} \biggr)^{\pm 1}
	   \frac{1}{1 - \bigl(\fa_0 (\la)/\fa_n (\la|\a)\bigr)^{\pm 1}} \\
        + \frac{\r_n^{\pm 1} (\la|\a)}
	       {1 - \bigl(\fa_n (\la|\a)/\fa_0 (\la)\bigr)^{\pm 1}} \Biggr] \epp
\end{multline}
Then we argued that the second term in the square brackets is
holomorphic inside the contours ${\cal C}_n$ at least for $T$
small enough. In fact, the functions $\r_n (\cdot|\a)^{\pm 1}$
are holomorphic inside ${\cal C}_n$. Moreover, 
\begin{equation} \label{repden}
     1 - \fa_0 (\la)/\fa_n (\la|\a)
        = 1 - \re^{2 \p \i (w (\tilde \la) + \a'' - 1/2)} + {\cal O} (T) \epc
\end{equation}
where $\tilde \la = \la - \i \g/2$ and $\la$ inside ${\cal C}_n$.
We have numerical evidence that the ${\cal O} (1)$ term on the
right hand side is nonzero inside ${\cal C}_n$, where it is
holomorphic as well. Since the latter is also true for the kernels,
we may replace the measures $\rd m^\a_\pm$ by the first terms
on the right hand side of (\ref{repweights}). This has the 
advantage that the new measures do not have poles that would
pinch the contour at the Fermi points when $T \rightarrow 0$.
Hence, we can shift the contour away from the Fermi points
and avoid the calculation of singular integrals. Starting from
(\ref{evratc0}) we then obtain
\begin{equation}
     \r_n (\la|\a) \frac{1 + \fa_0 (\la)}{1 + \fa_n (\la|\a)} =
        \re^{\i \p \a' + E(Q - \tilde \la) 
	      + \int_{-Q}^Q \rd \m \: \re(\m - \tilde \la) (w(\m) + \a'' - 1/2)}
	      + {\cal O} (T) \epc
\end{equation}
where $\la = \tilde \la + \i \g/2$ is outside ${\cal C}_n$ and
uniformly away from the Fermi points at $\pm Q + \i \g/2$.

Let us define
\begin{multline} \label{defcapms}
     \rd \hat M^\a_\pm (\la) = \frac{\rd \la}{2 \p \i}
        \frac{q^{\pm \a} \re^{\pm \int_{-Q}^Q \rd \m \: \re(\m - \la) (w(\m) - w(\la))}}
             {1 - \re^{\pm 2 \p \i (w (\la) + \a'' - s/2)}} \\ \times
        \re^{\pm \{ (w(\la) + \a'' + s/2) E(Q - \la)
	              - (w(\la) + \a'' - s/2) E(-Q - \la) \}} \epp
\end{multline}
Then, if we perform the zero temperature limit and shift the Fermi points
down to the real axis, we can replace the measures $\rd m^\a_\pm$ 
in the Fredholm determinants in the numerator in (\ref{defdl}) by
$\rd \hat M^\a_\pm$ with $s = 1$. In (\ref{defcapms}) we have also
separated the factors which are singular at the Fermi points from
the regular factors. Subsequently, assuming that there are no
singularities between the lower part of the integration contour
and the real axis, we deform the integration contour into a narrow
contour $\G [-Q, Q]$ encircling the interval $[-Q,Q]$ in positive
direction (see Figure~\ref{fig:gqqcontour}). Finally,
\begin{figure}[t]
\begin{center}
\includegraphics[width=.6\textwidth]{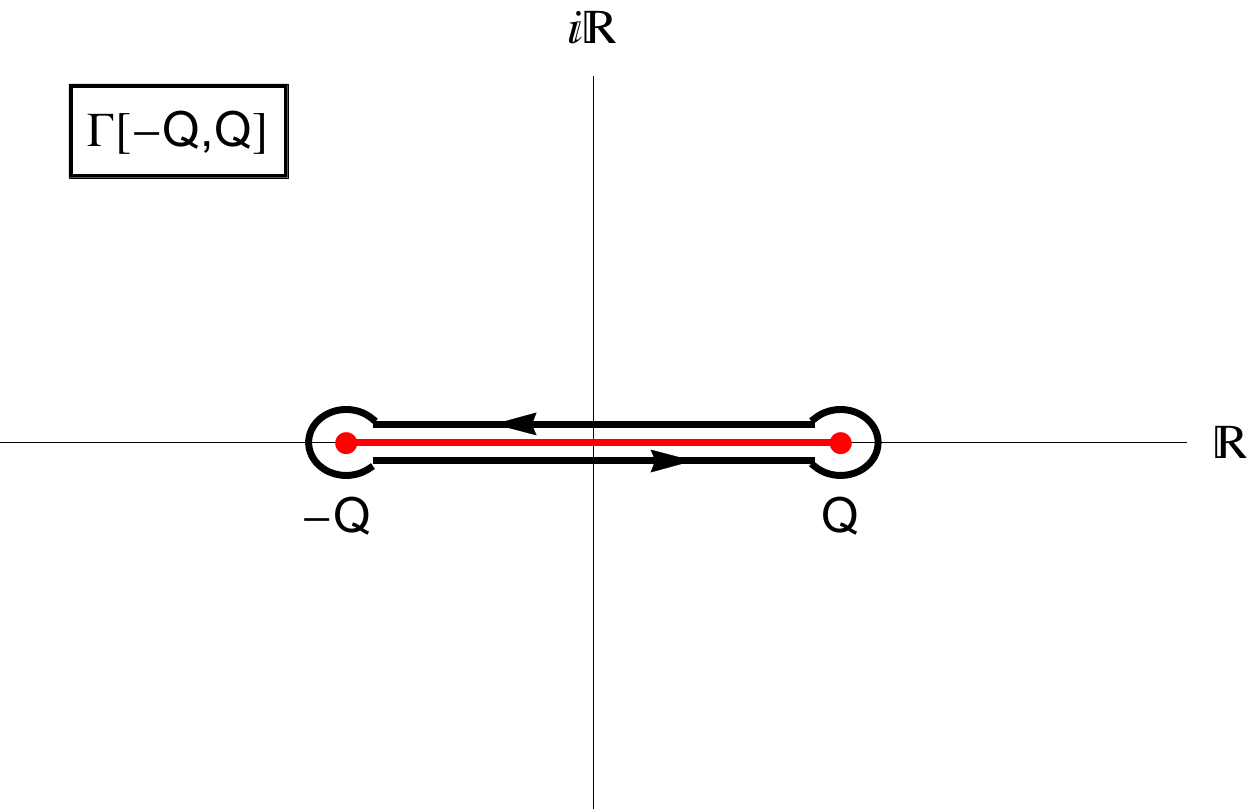}
\caption{\label{fig:gqqcontour}
Integration contour $\G[-Q, Q]$ involved in the Fredholm determinants
in the numerator of equation~(\ref{fredzero}).
}
\end{center}
\end{figure}%
\begin{equation} \label{fredzero}
     {\cal D} (p + \ell) =
        \frac{\det_{\rd \hat M^\a_+, \G[-Q, Q]} \bigl\{ 1 - \widehat{K}_{1-\a} \bigr\}
	      \det_{\rd \hat M^\a_-, \G[-Q, Q]} \bigl\{ 1 - \widehat{K}_{1+\a} \bigr\}}
	     {\det^2_{\rd \la/2 \p \i, [-Q, Q]} \bigl\{ 1 - \widehat{K} \bigr\}} \epp
\end{equation}

Here several remarks are in order. First, for the holomorphicity of the
second terms on the right hand side of (\ref{repweights}) inside
${\cal C}_n$ as well as for the contraction of the contour leading
to $\G [-Q,Q]$ we needed that the right hand side of (\ref{repden})
is non-zero inside ${\cal C}_n$. We verified this numerically with
examples, but it should be justified more rigorously e.g.\ by
establishing bounds on the imaginary part of $w$. Second, for
$\a = 0$ and $|p + \ell|$ large enough there may appear zeros
of the right hand side of (\ref{repden}) in the interval $[- Q,Q]$.
These are by definition outside the contour $\G[-Q,Q]$.

Our third remark concerns the discontinuity of the measures across the
interval $[-Q, Q]$. Taking into account the branch cut of $E(- Q - \la)$
along the real axis from $- Q$ to~$+ \infty$ we obtain
\begin{multline} \label{measuredelta}
     \rd \D_\pm^\a (\la) = \rd \hat M_\pm^\a (\la_-) - \rd \hat M_\pm^\a (\la_+) \\
        = \frac{\rd \la}{2 \p \i}
          q^{\pm \a} \, \re^{\pm \int_{-Q}^Q \rd \m \: \re(\m - \la) (w(\m) - w(\la))} \;
          \re^{\pm \{ (w(\la) + \a'' + s/2) E(Q - \la)
	              - (w(\la) + \a'' - s/2) E(-Q - \la_-) \}} \epp
\end{multline}
Hence, for a small $\e > 0$, we can interpret the integral over
$\G[-Q,Q]$ with measures $\rd \hat M_\pm^\a$ as a sum of an integral
over $[- Q + \e, Q -\e]$ with measures $\D_\pm^\a$ and two integrals
over infinitesimal circles of radius $\e$ around $- Q$ and $Q$  with
measures $\rd \hat M_\pm^\a$ (see Figure~\ref{fig:gqqcontour}). In
general the individual contributions do not exist in the limit
$\e \rightarrow 0$, because of the singularities of the measures at
the Fermi points. For $s = 1$ these are determined by the exponents
\begin{equation}
     w(\pm Q) + \a'' \pm 1/2 = \a'' {\cal Z} \pm \frac1{2 \cal Z} \epp
\end{equation}
In the special case $\a = p + \ell = 0$ which determines the leading
low-temperature asymptotics (see below), however,
\begin{equation}
     |w(\pm Q) + \a'' \pm 1/2| = \frac1{2 \cal Z} < \frac{1}{\sqrt{2}} \epc
\end{equation}
and the singularities of the measure are integrable. In this case we
may neglect the integrals over the small circles of radius $\e$ and
replace $\rd \hat M_\pm^\a$ by $\rd \D_\pm^\a$ and $\G [- Q, Q]$ by
$[- Q, Q]$ in the Fredholm determinants in the numerator of
(\ref{fredzero}).

Using similar ideas as above we can also obtain the zero temperature
form of the integral equations (\ref{liegbar}). But when we insert
(\ref{repweights}) into the integrals in (\ref{liegbar}) we have
to take into account that the functions $\overline{G}_\pm (\cdot, \x)$
are meromorphic with a single simple pole with residue $-1$ at
$\la = \x$. For this reason the second terms on the right hand side
of (\ref{repweights}) cannot be neglected. They contribute to the
driving term in the zero-temperature form of the integral equation.
Another contribution is obtained when we contract the integration
contours. Finally in the zero temperature limit the functions
$\overline{G}_\pm (\cdot, \x)$ are determined by the integral
equations
\begin{multline} \label{liegbarzero}
     \overline{G}_\pm (\la, \x) = - \cth(\la - \x) \\
        + q^{- \a \pm 1} \cth(\la - \x + \h) q^{\pm \a}
        \re^{\pm \{E(Q - \x + \i \g/2) 
	     + \int_{-Q}^Q \rd \m \: \re(\m - \x + \i \g/2) (w(\m) + \a'' - 1/2)\}}
	     \\[1ex]
	+ \int_{\G[-Q,Q]} \rd \hat M_\pm^\a (\m) \: \overline{G}_\pm (\m + \i \g/2, \x)
	  K_{\a \mp 1} (\m + \i \g/2 - \la) \epp
\end{multline}
Clearly the solutions depend only on $p + \ell$. For the physical
correlation functions we will later set $\x = 0$. In this case the
exponential contribution to the driving term simplifies,
\begin{equation}
     \re^{\pm \{E(Q + \i \g/2) 
          + \int_{-Q}^Q \rd \m \: \re(\m + \i \g/2) (w(\m) + \a'' - 1/2)\}}
        = - \re^{\mp 2 \i \a'' k_F} \epp
\end{equation}

\subsection{Summation}
We now turn to the summation over all excitations close to the Fermi
points (in the sense of (\ref{abovefermi})). An explicit  summation over
the particle and hole quantum numbers is possible for each value of $\ell$.
We may use the same summation formula as employed in \cite{KKMST11b}
in the context of the so-called critical form factors pertaining to
the eigenstates of the ordinary transfer matrix. This formula, adapted
to our notation, takes the form
\begin{multline} \label{critsum}
     \sum_{\substack{n_p, n_h \ge 0\\ n_p - n_h = \ell}}
     \sum_{\substack{p_1 < \dots < p_{n_p}\\ p_j \in {\mathbb Z}_+}}
     \sum_{\substack{h_1 < \dots < h_{n_h}\\ h_j \in {\mathbb Z}_+}} \mspace{-6.mu}
     \re^{- \frac{2 \p m T}{v_0}
          \bigl[ \sum_{j=1}^{n_p} (p_j - 1) + \sum_{j=1}^{n_h} h_j \bigr]}
	  \biggl( \frac{\sin(\p v)}{\p} \biggr)^{2 n_h} \mspace{-12.mu}
	  {\cal R}_{n_h, n_p} \bigl( \{h_j\}, \{p_j\} \big| v \bigr) \\[-2.ex]
        = \frac{G^2 (1 + \ell - v)}{G^2 (1 - v)} \,
          \frac{\re^{- \frac{\p m T \ell (\ell -1)}{v_0}}}
	       {\bigl( 1 - \re^{- \frac{2 \p m T}{v_0}} \bigr)^{(\ell - v)^2}} \epp
\end{multline}
It can be directly applied to the results of the previous subsections.
The terms that remain after the summation over the particle and
hole quantum numbers depend on $p$ and $\ell$ only through their sum.
For this reason we can shift the index and remain with a sum over
$\ell$. Performing also the limit $\a \rightarrow 0$ and setting
$\x = 0$ we obtain the following result:
\begin{equation} \label{osc}
     \<\s_1^- \s_{m+1}^+\>_{\rm osc}
        = (-1)^m \sum_{\ell \in {\mathbb Z}} A_{0, \ell}^{-+} \re^{2i m \ell k_F}
                 \biggl( \frac{\p T/v_0}{\sh(\p m T/v_0)}
		         \biggr)^{2 \ell^2 {\cal Z}^2 + \frac{1}{2 {\cal Z}^2}} \epc
\end{equation}
where
\begin{multline} \label{oscamp}
     A_{0, \ell}^{-+} =
        \frac{\re^{\2 \th(2Q) + E(-2Q) + C_- [w] + C_+ [w]}}{4 \g \sh(\h)}
	{\cal D}(\ell) \, \overline{G}_+^- (0) \,
	\6_\a \overline{G}_-^+ (0) \big|_{\a = 0} \,
	\re^{\ell \th (2Q) \bigl(\frac{\cal Z}{2} - \frac{1}{2 \cal Z}\bigr)} \\[1ex]
        \times \biggl[ \prod_{\e_1, \e_2 = \pm 1}
	               G \biggl(1 + \e_1 \ell {\cal Z}
		                  + \frac{\e_2}{2 {\cal Z}} \biggr) \biggr]
        \biggl( \frac{\re^{\2 \th(2Q) - E(2Q)}}{2 \p \r (Q) \sh(\h)}
	        \biggr)^{2 \ell^2 {\cal Z}^2 + \frac{1}{2 {\cal Z}^2}} \epp
\end{multline}
and where it is understood that $\a = p = 0$ in those terms which depend
implicitly on $\a$ and $p$.

The series (\ref{osc}) is not an asymptotic series, neither in $T$ nor in
$m$. In each order of exponential decay we have neglected algebraic corrections
in $T$ that would contribute lower order terms than the next-order exponentials.
Moreover, we have neglected higher temperature corrections to the
correlation lengths that would contribute terms of the form $\exp{{\cal O}
(m T^2)}$. The series (\ref{osc}) is systematic in that it gives the
leading amplitudes in front of every oscillating term $\re^{2i m \ell k_F}$.

The leading low-temperature large-distance asymptotics of the transversal
correlation functions is given by the $\ell = 0$ term in the sum, \textit{i.e.}
\begin{multline} \label{transverseasymp}
     \<\s_1^- \s_{m+1}^+\> \sim
          \frac{\re^{\2 \th(2Q) + E(-2Q) + C_- [w] + C_+ [w]}}{4 \g \sh(\h)}
	  {\cal D}(0) \, \overline{G}_+^- (0)\,
	  \6_\a \overline{G}_-^+ (0) \big|_{\a = 0} \\[1ex] \times
          G^2 \Bigl(1 + \frac{1}{2 {\cal Z}}\Bigr)
	  G^2 \Bigl(1 - \frac{1}{2 {\cal Z}}\Bigr)
          \biggl( \frac{\re^{\2 \th (2Q) - E(2Q)}}{2 \p \r (Q) \sh(\h)}
		  \biggr)^{\frac{1}{2 {\cal Z}^2}} (-1)^m
          \biggl( \frac{\p T/v_0}{\sh(\p m T/v_0)} \biggr)^{\frac{1}{2 {\cal Z}^2}} \epp
\end{multline}
This is our main result. We shall see below that this formula is numerically
efficient and matches well with known results.
\section{Discussion}
In our previous work \cite{DGK13a} in which we derived our formulae for
the amplitudes we also analyzed a generating function of the longitudinal
correlation functions for small temperatures. In that work we omitted
a detailed discussion of the meaning of $p$ which was supplemented here.
We also postponed the numerical analysis of the longitudinal case.
Before we catch up on this let us recall the formulae.
\subsection{A summary of the longitudinal case}
In \cite{DGK13a} we obtained an `oscillating series' of similar
form and meaning as (\ref{osc}) for a generating function of the
longitudinal correlation functions,
\begin{equation} \label{osczz}
     \< \re^{2 \p \i \a S(m)} \>_{\rm osc}
        = (-1)^{m \a} \sum_{\ell \in {\mathbb Z}} A_{0, \ell} \,
	  \re^{2i m (\ell - \a) k_F} \biggl( \frac{\p T/v_0}{\sh(\p m T/v_0)}
		                             \biggr)^{2 (\ell - \a)^2 {\cal Z}^2} \epp
\end{equation}
Here $S(m) = \sum_{j=1}^m \s_j^z/2$. The amplitudes consist of two
factors, $A_{0, \ell} = {\cal D}_z (\ell) {\cal A} (\ell - \a)$, where
\begin{equation}
     {\cal A} (x) = \re^{C_z [x Z]} G^2 (1 + x {\cal Z}) G^2 (1 - x {\cal Z})
        \biggl( \frac{\re^{\2 \th (2Q) - E(2Q)}}{2 \p \r (Q) \sh(\h)}
	        \biggr)^{2 x^2 {\cal Z}^2}
\end{equation}
with
\begin{multline}
     C_z [v] = \2 \int_{- Q}^Q \rd \la \int_{- Q}^Q \rd \m \:
	       \bigl( v' (\la) v (\m) - v (\la) v' (\m) \bigr) \re (\la - \m) \\
               + 2 v (Q) \int_{- Q}^Q \rd \la
	         \bigl(v (\la) -  v (Q)\bigr) \re (\la - Q) \epp
\end{multline}
The other factor stems from the Fredholm determinant part of the
amplitudes. It is similar to (\ref{fredzero}) and can be written as
\begin{equation} \label{ratdetzero}
     {\cal D}_z (\ell) = 
        \frac{\det_{\rd \hat{M}^\a_+, \G [- Q, Q]}
	      \bigl\{ 1 - \widehat{\cal K}_{- \a} \bigr\}
	      \det_{\rd \hat{M}^\a_-, \G [- Q, Q]}
	      \bigl\{ 1 - \widehat{\cal K}_{\a} \bigr\}}
	     {\det^2_{\frac{\rd \la}{2\p \i}, [-Q, Q]}
	      \bigl\{ 1 - \widehat{K} \bigr\}} \epp
\end{equation}
Here it is understood that $s = 0$ in the measures (cf.\ (\ref{defcapms}))
in the numerator. In this case it is convenient to absorb the factorized
part of the form factors into the Fredholm determinant part \cite{DGK13a},
leading to a modification of the kernel in the Fredholm determinants
in the numerator,
\begin{equation} \label{kernela}
     {\cal K}_\a (\la) = \frac{\re^{(\a - 1)(\la - \h)}}{\sh(\la - \h)}
                         - \frac{\re^{(\a - 1)(\la + \h)}}{\sh(\la + \h)} \epp
\end{equation}

Starting from the above oscillating series for the generating
function we calculated the leading asymptotics of the longitudinal
correlation functions in \cite{DGK13a}. Taking into account
only the terms with $\ell = -1, 0 , 1$ in (\ref{osczz}) we obtained
\begin{subequations}
\label{longitudinalasymp}
\begin{align}
     & \<\s_1^z \s_{m+1}^z\> - \<\s_1^z\> \<\s_{m+1}^z\> \sim
          \notag \\ & \mspace{108.mu}
	  A^{zz}_{0,0} \biggl( \frac{\p T/v_0}{\sh( m \p T/v_0)} \biggr)^2 + \:
	  A^{zz}_{0,1} \cos(2m k_F)
	  \biggl( \frac{\p T/v_0}{\sh( m \p T/v_0)} \biggr)^{2 {\cal Z}^2}
	  \mspace{-18.mu} \epc \\[1ex]
     & A^{zz}_{0,0} = - \frac{2 {\cal Z}^2}{\p^2} \epc \qd
       A^{zz}_{0,1} = \frac{4 \sin^2 (k_F)}{\p^2}
                      {\cal A} (1) {\cal D}_z'' (1) \epp
\end{align}
\end{subequations}
Here we have used the shorthand notation ${\cal D}_z'' (1) = \6_\a^2
{\cal D}_z (1)|_{\a = 0}$. This $\a$-derivative can be calculated
analytically (see e.g.\ Appendix D of \cite{DGK13a}).

\subsection{Numerical evaluation and comparison with known results}
We would like to stress that the asymptotic formulae
(\ref{transverseasymp}) and (\ref{longitudinalasymp}) are numerically
efficient and can be evaluated with standard software on a laptop
computer. In fact, what has to be calculated are basically the
solutions of the linear integral equations (\ref{zphr}) and
(\ref{rho}) which then have to be integrated over or evaluated
at the Fermi points. Moreover, we have to solve the linear
integral equations (\ref{liegbarzero}) and have to calculate the
Fredholm determinants in (\ref{fredzero}). As we have learned
from \cite{Bornemann10} Fredholm determinants can be efficiently
calculated be discretization. The only additional problem
we encounter in the Fredholm determinants in the numerator of
equation (\ref{fredzero}) and also in the linear integral
equations (\ref{liegbarzero}) is the weakly singular behaviour
of the integration measures $\rd \D_\pm^\a$, equation
(\ref{measuredelta}). It can be dealt with by means of standard
Gau{\ss}-Jacobi quadrature.
\begin{figure}[t]
\begin{tabular}{@{}ll@{}}
\includegraphics[width=.48\textwidth]{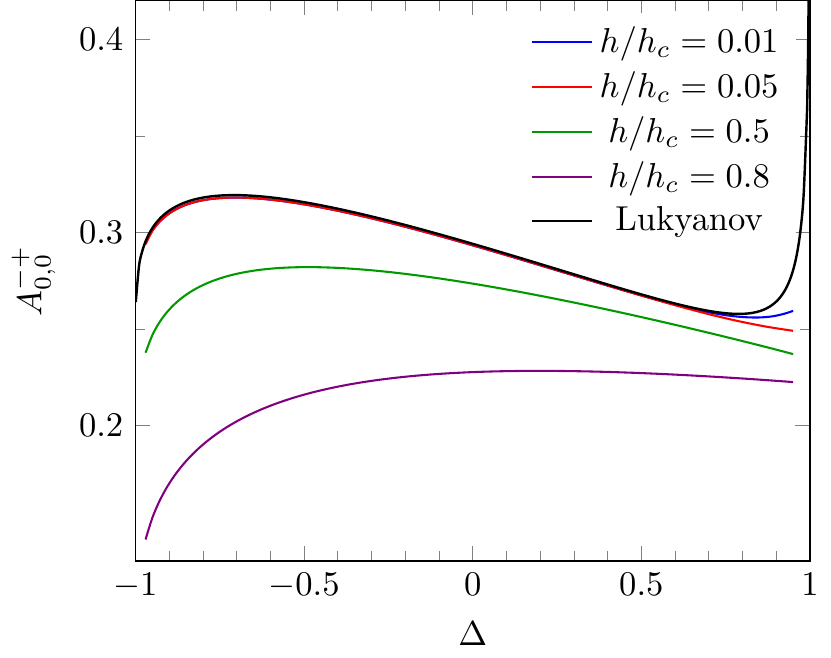} &
\includegraphics[width=.48\textwidth]{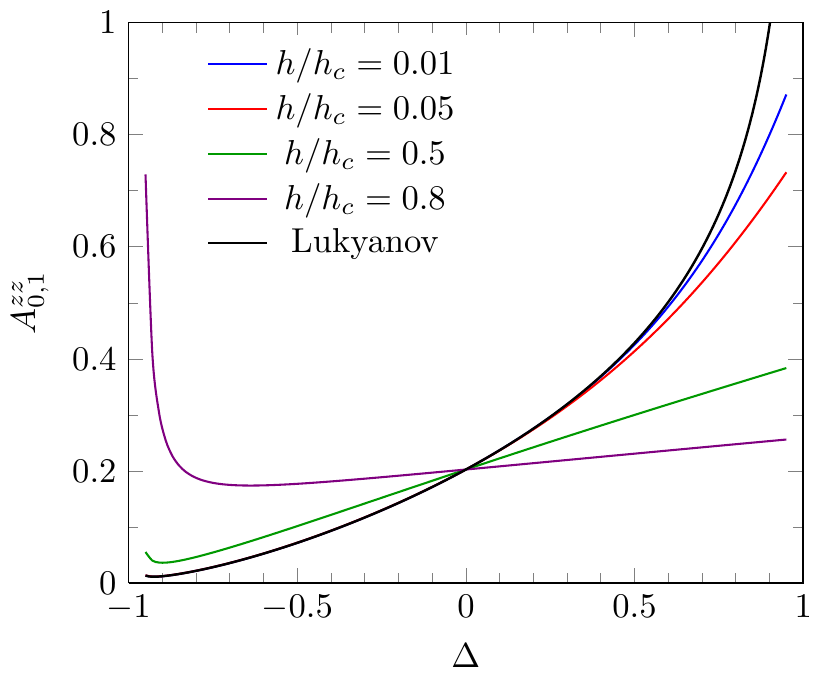} \\
\includegraphics[width=.48\textwidth]{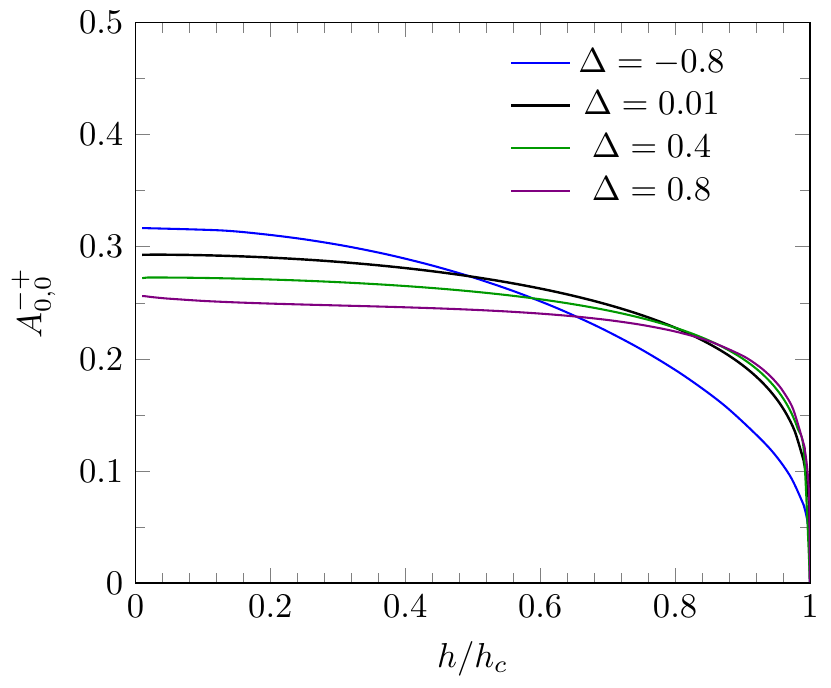} &
\includegraphics[width=.48\textwidth]{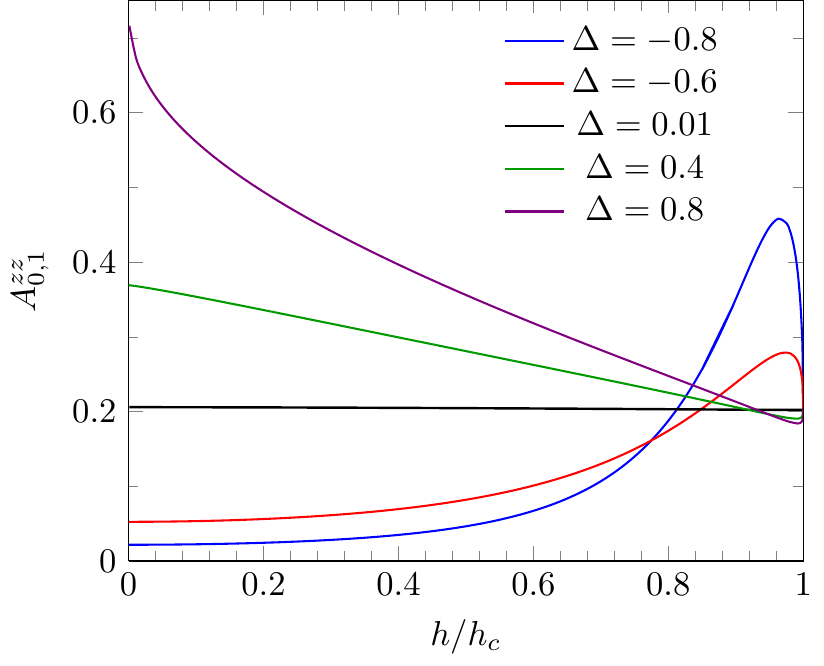}
\end{tabular}
\caption{\label{fig:ampsvsdelta} Amplitudes in the leading asymptotic
terms as functions of the anisotropy parameter for various values
of the magnetic field (upper panels) and as functions of the
magnetic field for various values of the anisotropy parameter
(lower panels). Transversal case in the left panels and longitudinal
case in the right panels. In the longitudinal case the amplitude
is leading only for $\D > 0$. For decreasing values of the magnetic
field we observe numerical convergence to the $h = 0$ result
\cite{Lukyanov98,Lukyanov99} of Lukyanov.
}%
\end{figure}%

In Figure \ref{fig:ampsvsdelta} we show the amplitudes $A^{zz}_{0,1}$
and $A^{-+}_{0,0}$ as functions of the anisotropy parameter $\D$
for various values of the magnetic field. Recall that our formulae
are valid for any positive magnetic field $h$. They are complementary
to the $h = 0$ results \cite{Lukyanov99}
\begin{align} \label{luckyamps}
     & A^{zz}_{0,1} = \frac{8}{\p^2}
                      \Biggl[ \frac{\G \bigl(\frac{\p}{2 \g} - \2\bigr)}
		                   {2 \sqrt{\p} \; \G \bigl(\frac{\p}{2 \g} \bigr)}
				           \Biggr]^{\frac{\p}{\p - \g}}
		      \mspace{-9.mu}
                      \exp \Biggl\{
		           \int_0^\infty \frac{\rd k}{k} \biggl[
			   \frac{\sh \bigl( (1 - \frac{2 \g}{\p})k \bigr)}
			        {\sh \bigl( (1 - \frac{\g}{\p})k \bigr)
				 \ch(\frac{\g k}{\p})} -
		           \Bigl(1 - \frac{\g}{\p - \g} \Bigr) \re^{- 2k} \biggr]
			   \Biggr\} \epc \notag \\[1ex]
     & A^{-+}_{0,0} = \frac{\p^2}{4 \g^2}
                      \Biggl[ \frac{\G \bigl(\frac{\p}{2 \g} - \2\bigr)}
		                   {2 \sqrt{\p} \; \G \bigl(\frac{\p}{2 \g} \bigr)}
				           \Biggr]^{1 - \frac{\g}{\p}}
                      \exp \Biggl\{
		           \int_0^\infty \frac{\rd k}{k} \biggl[
		           \Bigl(1 - \frac{\g}{\p} \Bigr) \re^{- 2k}
			   - \frac{\sh \bigl( (1 - \frac{\g}{\p})k \bigr)}
			          {\sh (k) \ch(\frac{\g k}{\p})} \biggr] \Biggr\}
\end{align}
obtained in a quantum field theoretic setting \cite{LuPe75,Cardy86,%
Lukyanov98,LuTe03} starting from the Gaussian model and taking into
account the most relevant irrelevant operators. The amplitudes
(\ref{luckyamps}) are plotted in black in the upper panels of
Figure~\ref{fig:ampsvsdelta}. Clearly, our field-dependent amplitudes
numerically converge to these amplitudes for $h \rightarrow 0$.
The discrepancies close to $\D = 1$ are not numerical artifacts.
They rather indicate the highly singular behaviour of the amplitudes
close to this point. So far we do not know how to obtain the
expressions (\ref{luckyamps}) directly from our formulae
(\ref{longitudinalasymp}) and (\ref{transverseasymp}). In
the lower panels of Figure~\ref{fig:ampsvsdelta} we show the
field dependence of the amplitudes. Their slope as functions
of $h$ becomes infinite at the critical field $h_c$, where
the phase transition to the fully polarized phase occurs.
\begin{figure}[t]
\begin{tabular}{@{}ll@{}}
\includegraphics[width=.48\textwidth]{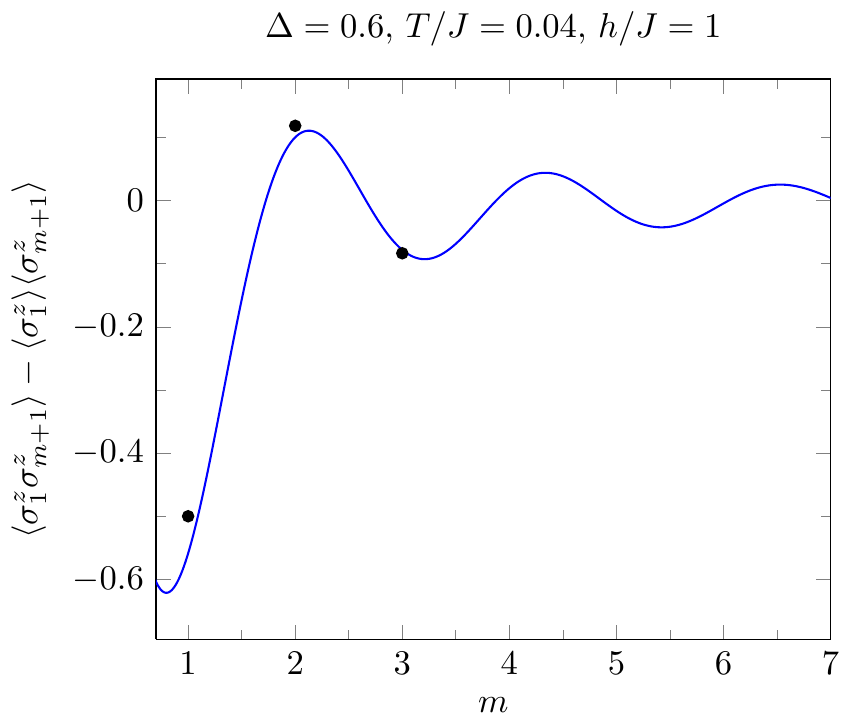} &
\includegraphics[width=.48\textwidth]{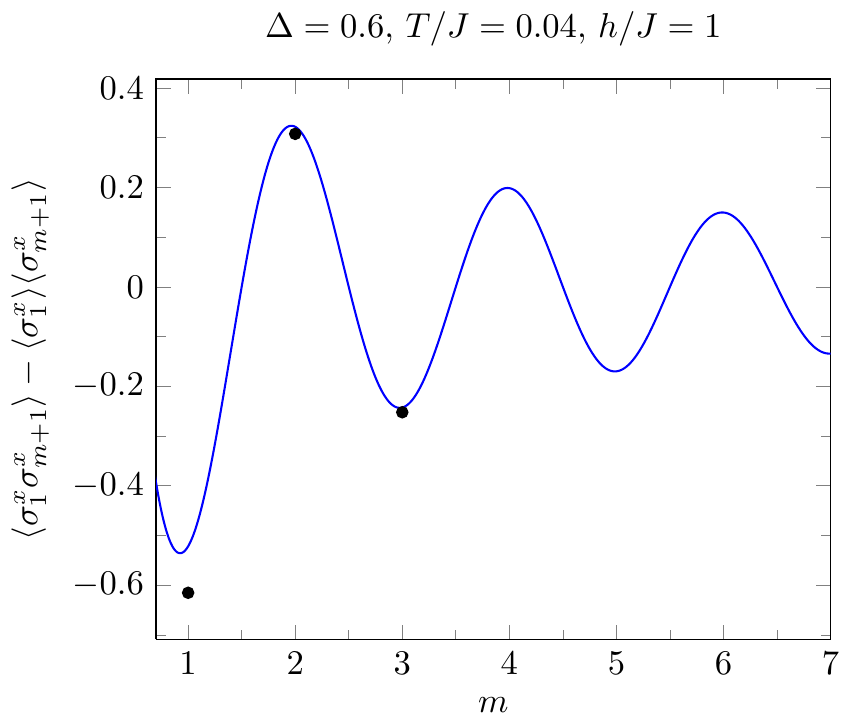}
\end{tabular}
\caption{\label{fig:corrvsdistance} 
Correlation functions according to our asymptotic formulae
(\ref{longitudinalasymp}) (left panel) and (\ref{transverseasymp})
(right panel). Black dots depict the exact values at short distances
obtained in \cite{BDGKSW08}.
}%
\end{figure}%

Knowing the amplitudes we can plot the correlation functions
as functions of the distance. This is shown in
Figure~\ref{fig:corrvsdistance}, where we also compare the
asymptotic behaviour with the exact short-distance behaviour
of the two-point functions as obtained in \cite{BDGKSW08}.
As we see, the asymptotic formulae provide rather accurate
approximations down to the smallest possible distance $m = 1$.
This seems less amazing if we take into account that fairly good
agreement between the leading order asymptotic formulae and
the exact short-distance results was even obtained at the
isotropic point, where logarithmic corrections are important
\cite{SABGKTT11}.

This encouraged us to compare the magnetic field dependence of
the third-neighbour correlation functions as obtained from
the asymptotic formulae with the exact results, which is
shown in Figure~\ref{fig:thirdneighbourh}. Again the
agreement is good. It becomes worse if we approach the
isotropic point. The temperature behaviour of the third-neighbour
longitudinal two-point function for two different values of
anisotropy is shown in Figure~\ref{fig:thirdneighbourhlongt}.
\begin{figure}[t]
\begin{tabular}{@{}ll@{}}
\includegraphics[width=.48\textwidth]{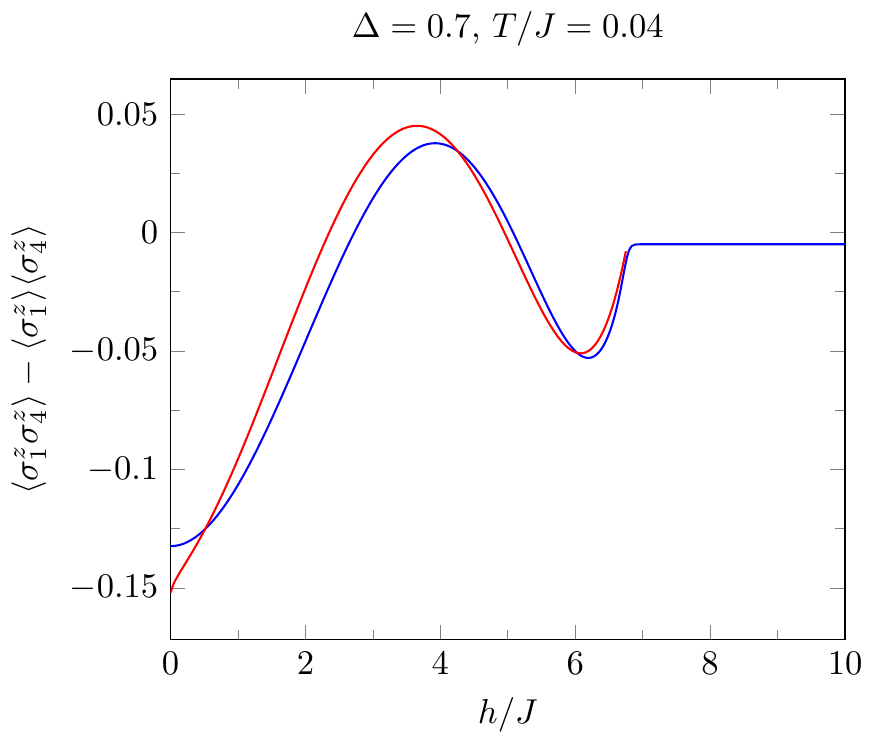} &
\includegraphics[width=.48\textwidth]{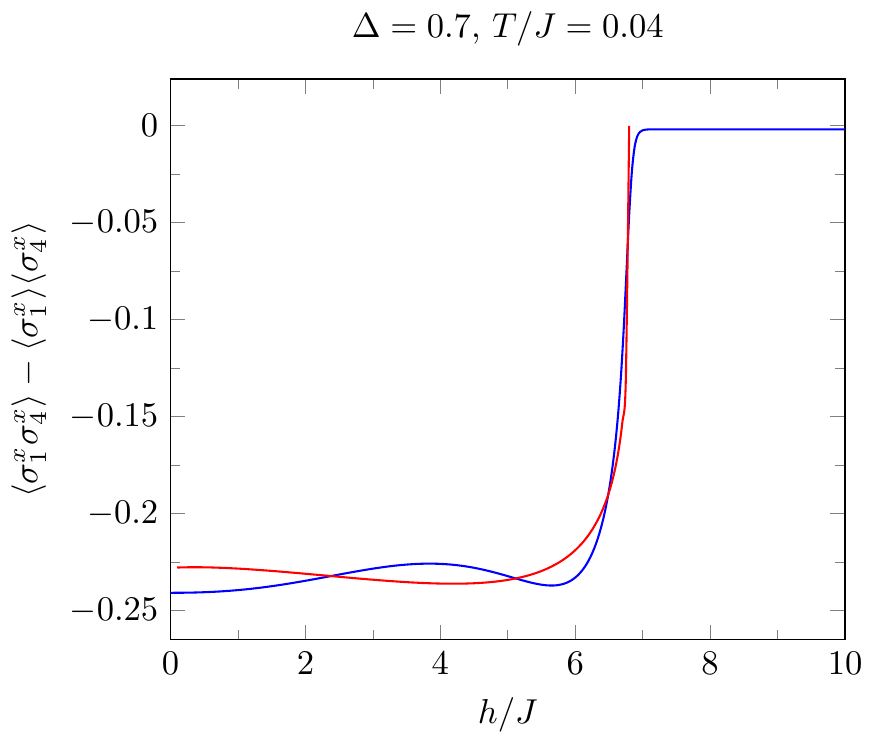}
\end{tabular}
\caption{\label{fig:thirdneighbourh} Comparison of magnetic
field dependence of the third-neighbour correlators, exact
\cite{BDGKSW08} blue lines, asymptotic red lines. Longitudinal
case in the left panel, transversal case in the right panel.
}%
\end{figure}%
\begin{figure}[t]
\begin{tabular}{@{}ll@{}}
\includegraphics[width=.48\textwidth]{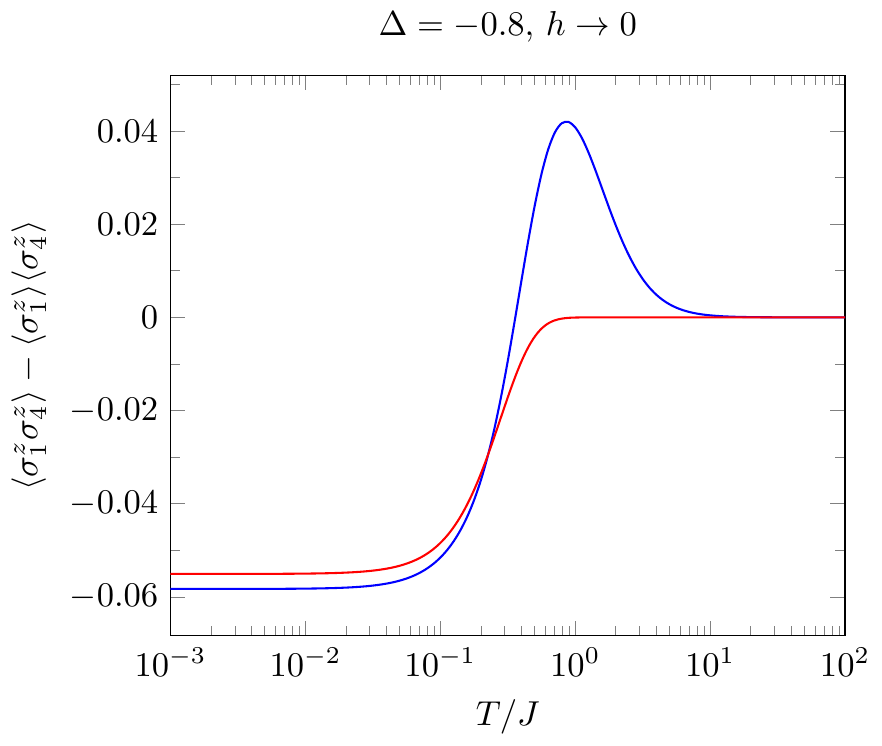} &
\includegraphics[width=.48\textwidth]{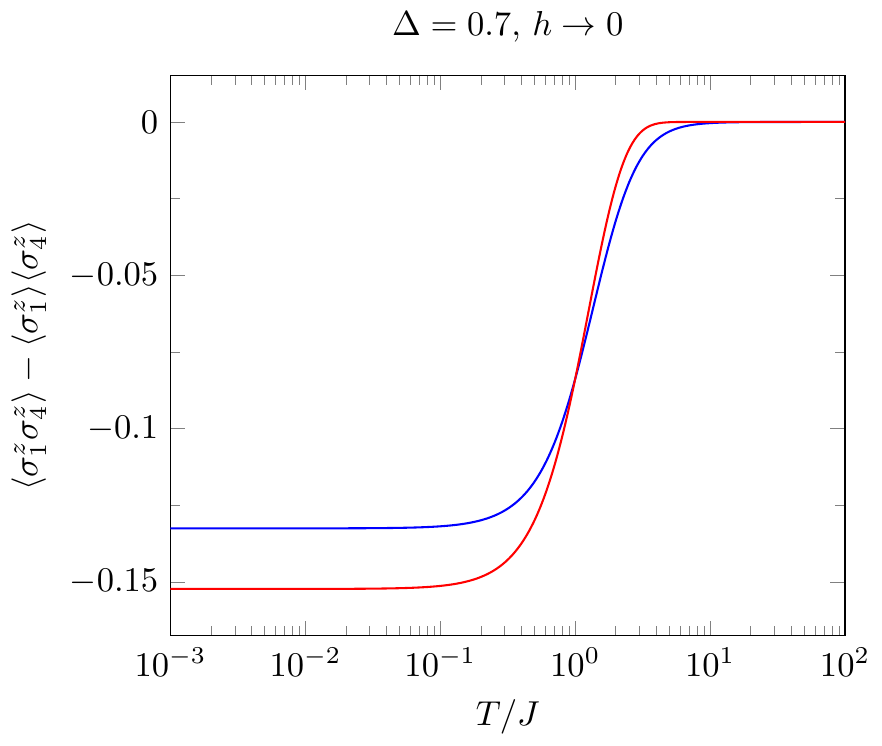}
\end{tabular}
\caption{\label{fig:thirdneighbourhlongt} Comparison of temperature
dependence of third-neighbour correlators, exact \cite{BDGKSW08}
blue lines, asymptotic red lines. The sign change of the correlation
function for negative $\D$ (blue line in left panel) has been explained
by a `quantum classical crossover' \cite{FaMc99,FKM99,BDGKSW08}.
Accordingly it is not seen in our low-temperature asymptotic
expansion which covers only the `quantum regime' of the phase
diagram.
}%
\end{figure}%
For negative $\D$, where these correlation functions exhibit a
`quantum classical crossover' \cite{FaMc99,FKM99,BDGKSW08},
the asymptotic result ceases to be a good approximation for
temperatures above the crossover temperature. For positve
$\D$, however, where no such crossover occurs, the asymptotic
formula provides a reasonably accurate description of the
correlation functions up to arbitrary temperatures.
\section{Conclusion}
In this work we have continued the low-temperature analysis of
the correlation lengths and amplitudes occurring in the form
factor expansion of the two-point correlation functions of the
spin-$1/2$ XXZ chain. We started our analysis with expressions
for the amplitudes which were obtained in \cite{DGK13a}, where
we combined algebraic Bethe ansatz methods for the calculation
of form factors with the quantum transfer matrix approach to
thermodynamics \cite{Suzuki85,SuIn87} and with the method of
nonlinear integral equations \cite{Kluemper92,Kluemper93}.
Similar formulae hold for finite-size systems at zero
temperature. In both cases the formulae are exact up to this
stage. Their low-temperature (or finite-size) analysis is a
logically independent task. Here we dealt with the low-temperature
case for $|\D| < 1$ and positive magnetic field. We provided,
in particular, an extensive discussion of the low-temperature behaviour
of the solution of the fundamental nonlinear integral equations.
Higher temperatures and the case $\D > 1$ will be addressed in
separate works.

Based on the low-temperature analysis of the nonlinear integral
equations we obtained the leading low-temperature expressions
for the correlation lengths and for the amplitudes associated
with those excitations that `collapse to the Fermi points' for $T
\rightarrow 0$. These amplitudes show `critical behaviour':
they vanish as fractional powers of the temperature with
critical exponents determined by the scaling dimensions
of the underlying conformal field theory. Using a summation
formula obtained in \cite{KOV93,Olshanski03,KKMST11b}, we
could sum up the corresponding terms in the form factor
series and obtained the leading low-temperature large-distance
asymptotics of the transversal two-point functions. A similar
asymptotic formula for the longitudinal case was obtained
in \cite{DGK13a}. Our formulae include explicit expressions
for the amplitudes for any positive magnetic field which
do not follow directly from conformal field theory and are
complementary to Lukyanov's formulae \cite{Lukyanov99}
for $h = 0$. Similar but different expressions for the
amplitudes were obtained previously in the context of scaling
analysis for large system size at $T = 0$ \cite{KKMST09b,%
KKMST11a,SPCI12}. Our formulae have turned out to be
numerically efficient. They can be evaluated on a laptop
computer and the resulting curves for amplitudes and correlation
functions match well with known results.

\smallskip \noindent {\bf Acknowledgment.}
The authors are grateful to A. Kl\"umper, J. Suzuki and A. Wei{\ss}e
for helpful discussions and encouragement. The numerical data for the
exact short-range correlation functions were generously provided
by M. Brockmann. MD and FG acknowledge financial support by the
Volkswagen Foundation and by the Deutsche Forschungsgemeinschaft
under grant number Go 825/7-1. KKK is supported by the CNRS. His work has
been partly financed by the grant PEPS-PTI `Asymptotique d'int\'egrales
multiples' and by a Burgundy region PARI 2013 FABER grant `Structures et
asymptotiques d'int\'egrales multiples'.

\clearpage

{\appendix
\Appendix{Two proofs}
\label{app:proofs}
In this appendix we would like to provide the proofs of
Lemma~\ref{lem:sing1} and Lemma~\ref{lem:sing2} that were left
out in the main text.
\subsection{Proof of Lemma \ref{lem:sing1}}
\label{app:prooflemma2}
\noindent
By the same reasoning as in the proof of Lemma~\ref{lem:sommerfeld} we
obtain for any $\la_\pm$ not on ${\cal C}_u$
\begin{align} \label{prepdom}
     \D I_u (\la_\pm) := & I_u (\la_\pm)
        + \int_{Q_-}^{Q_+} \rd \la \: \cth(\la - \la_\pm)
	                   \frac{\uq (\la)}T \notag \\[1ex]
          = & \int_{{\cal C}_u} \rd \la \: \cth(\la - \la_\pm)
	                      \ln \Bigl(1 + \re^{- \frac{u(\la) \sign (v(\la))}T} \Bigr)
			      \notag \\[1ex]
          = & \int_{J_-^\de \cup J_+^\de} \rd \la \:
                  \cth(\la - \la_\pm) \ln \Bigl(1 + \re^{- \frac{|v(\la)|}T} \Bigr)
            + \CO (T^\infty) \epc
\end{align}
where $v(\la) = \Re u (\la)$, and $\de > 0$ and $J_\pm^\de$ are chosen as
in the proof of Lemma~\ref{lem:sommerfeld}.

If $\la_\pm \in V_+$, then
\begin{align} \label{domplus}
     \int_{J_-^\de \cup J_+^\de} \rd \la \: &
         \cth(\la - \la_\pm) \ln \Bigl(1 + \re^{- \frac{|v(\la)|}T} \Bigr)
	 \notag \\[1ex]
     & = \int_{J_+^\de} \rd \la \:
         \cth(\la - \la_\pm) \ln \Bigl(1 + \re^{- \frac{|v(\la)|}T} \Bigr)
         + \CO (T) \notag \\[1ex]
     & = \int_{J_+^\de} \rd \la \: \biggl[
         \cth(\la - \la_\pm) - \frac{u'(\la)}{u(\la) - u(\la_\pm)} \biggr]
	 \ln \Bigl(1 + \re^{- \frac{|v(\la)|}T} \Bigr) \notag \\
     & \mspace{72.mu} + \int_{J_+^\de} \rd \la \:
         \frac{v'(\la)}{v(\la) + \i w(\la) - u(\la_\pm)}
	 \ln \Bigl(1 + \re^{- \frac{|v(\la)|}T} \Bigr) + \CO (T) \notag \\[1ex]
     & = \int_{- \de/T}^{\de/T} \rd x \:
         \frac{\ln \bigl(1 + \re^{- |x|} \bigr)}{x - \uq (\la_\pm)/T}
	 + \CO (T) \epp
\end{align}
Here we used (\ref{correctminus}) in the first equation and a similar
identity with $J_+^\de$ replacing $J_-^\de$ in the third equation. We
further employed the fact that $w (\la) = 2 \p \i T$ on $J_+^\de$ in
the second and third equation. As in the proof of Lemma~\ref{lem:sommerfeld}
we substituted $x = v(\la)/T$.

For $\la_\pm$ in the vicinity of $Q_+$ we cannot assume that $a_\pm
:= \uq (\la_\pm)/T$ is large for small $T$. We therefore treat $a_\pm$
as an independent parameter. Then (\ref{prepdom}) and (\ref{domplus})
imply that
\begin{equation}
     \D I_u (\la_\pm) = \int_{- \infty}^\infty \rd x \:
         \frac{\ln \bigl(1 + \re^{- |x|} \bigr)}{x - a_\pm} + \CO (T) \epp
\end{equation}
Note that $a_+$ is in the upper half plane and $a_-$ is in the lower
half plane. The integral on the right hand side can be calculated by
means of the residue theorem (for a similar calculation in the
context of the Bose gas with delta function interaction see
\cite{KMS11b}),
\begin{multline}
     \int_{- \infty}^\infty \rd x \:
        \frac{\ln \bigl(1 + \re^{- |x|} \bigr)}{x - a_\pm} \\
	= \mp 2 \p \i \ln \Bigl\{ \G \Bigl( \2 \pm
	                          \frac{a_\pm}{2 \p \i} \Bigr) \Bigr\}
          \pm \p \i \ln (2 \p)
	  + a_\pm \ln \Bigl( \frac{a_\pm}{\pm 2 \p \i} \Bigr) - a_\pm \epp
\end{multline}
On the other hand
\begin{multline} \label{subtractuqlapm}
     \int_{Q_-}^{Q_+} \rd \la \: \cth(\la - \la_\pm) \frac{\uq (\la)}T \\
        = \int_{Q_-}^{Q_+} \rd \la \: \cth(\la - \la_\pm)
	     \frac{\uq (\la) - \uq(\la_\pm)}T
	     + \frac{\uq(\la_\pm)}T
	       \ln \biggl(\frac{\sh (Q_+ - \la_\pm)}{\sh (Q_- - \la_\pm)} \biggr)
	       \epc
\end{multline}
and (\ref{iuright}) follows from (\ref{prepdom})-(\ref{subtractuqlapm}).
The proof of (\ref{iuleft}) is similar and is left to the reader.
\subsection{Proof of Lemma \ref{lem:sing2}}
\label{app:prooflemma3}
\noindent
In preparation of the proof we shall need some results which are
either easy to see or were proved elsewhere. First of all we have
the following corollary of Lemma~\ref{lem:sommerfeld} and
Lemma~\ref{lem:sing1}.
\begin{corollary} \label{cor:zcauchy}
Let $\la \in V_\pm$ and $\s = \sign (\Im \la )$. Then $\la$ is above
${\cal C}_{0,s} - \i \g/2$ if $\s > 0$, below ${\cal C}_{0,s} - \i \g/2$
if $\s < 0$, provided that $T$ is small enough, and
\begin{multline} \label{closetoqpm}
     L_{{\cal C}_{0,s} - \i \g/2} [z] (\la) =
     \int_{{\cal C}_{0,s} - \i \g/2} \rd \m \: \cth(\m - \la) z(\m + \i \g/2) \\[1ex]
        = - \int_{-Q}^Q \frac{\rd \m}{2 \p \i} \:
	      \cth(\m - \la) \bigl(\uq_1 (\m) - \uq_1 (\la)\bigr)
	      + \frac{\uq_1 (\la)}{2 \p \i}
	        \ln \biggl( \frac{\e^{\pm 1} (\la) \sh(\la + Q)}{\sh(\la - Q)} \biggr) \\
	      \mp \frac{\uq_1 (\la)}{2 \p \i} \ln( \pm \s 2 \p \i T)
	      - \s \ln \biggl( \frac{\G \bigl(1/2 \pm \s \uq (\la)/2 \p \i T\bigr)}
	                           {\G \bigl(1/2 \pm \s \uq_0 (\la)/2 \p \i T\bigr)}
				   \biggr) + \CO (T) \epp
\end{multline}
On the other hand, if $\la$ is uniformly away from $\pm Q$, it follows that
\begin{equation}
     L_{{\cal C}_{0,s} - \i \g/2} [z] (\la) =
        - \int_{- Q}^Q \frac{\rd \m}{2 \p \i} \: \cth(\m - \la) \uq_1 (\m) + \CO (T) \epp
\end{equation}
\end{corollary}
The following useful lemma was proved in Appendix~B of \cite{KKMST09c}.
\begin{lemma} \label{lem:limint}
Let $I \subset {\mathbb R}$ be an open interval containing $0$.
Let ${\cal R}: I \times {\mathbb R}_+ \rightarrow {\mathbb C}$ such
that $u \rightarrow {\cal R} (u, t)$ is ${\cal C}^1 (I)$ for all
but finitely many $t$ and $t \rightarrow {\cal R} (u, t)$ is
Riemann integrable uniformly in $u$, \textit{i.e.}, $\forall\ \e > 0,
\forall\ M > 0, \forall\ u_0 \in I\ \exists\ \de > 0$ such that
$u \in (u_0 - \de, u_0 + \de) \cap I\ \then$
\begin{equation}
     \biggl| \int_M^\infty \rd t \bigl( \6_1^k {\cal R} (u, t)
                - \6_1^k {\cal R} (u_0, t) \bigr) \biggr| < \e
\end{equation}
for $k = 0, 1$.

Then, for $g \in {\cal C}^1 (I)$,
\begin{equation} \label{limint}
     \int_0^\de \rd t \: x g(t) {\cal R} (t, xt)
        = g(0) \int_0^\infty \rd t \: {\cal R} (0, t) + o(1) \epc
\end{equation}
where $o(1)$ denotes terms that vanish in the ordered limit
first $x \de \rightarrow \infty$ and then $\de \rightarrow 0$.
\end{lemma}

Now, for the proof of Lemma~\ref{lem:sing2}, we rewrite the double
integral as
\begin{multline}
    A = - \int_{{\cal C}_{0,s}} \rd \la \: \int_{{\cal C}_{0,s}'} \rd \m \:
          z(\la) \cth' (\la - \m) z(\m) = \\
        \int_{- ({\cal C}_{0,s} - \i \g/2)} \frac{\rd \la}{2 \p \i T} \:
	   \biggl[ \frac{u'(\la)}{1 + \re^{u(\la)/T}} -
	           \frac{u_0'(\la)}{1 + \re^{u_0(\la)/T}} \biggr]
           L_{{\cal C}_{0,s}' - \i \g/2} [z] (\la) \epp
\end{multline}

Then we deform the contour $- ({\cal C}_{0,s} - \i \g/2)$ and decompose
it into four disjoint pieces, $- ({\cal C}_{0,s} - \i \g/2) \rightarrow
{\cal C}_+^\de \cup {\cal C}_-^\de \cup J_+^\de \cup J_-^\de$ (see Figure~%
\ref{fig:contour_decomposition}). Here $\de > 0$, and $J_\pm^\de$ are
chosen as in the proof of Lemma~\ref{lem:sommerfeld}. They contain $Q_\pm$,
$\Im u = 2 \p p T$ in $J_\pm^\de$ and $\Re u$ grows (decreases) monotonically
from $- \de$ to $\de$ in $J_+^\de$ ($J_-^\de$). Moreover, $\Re u > 0$ on
${\cal C}_+^\de$, $\Re u < 0$ on ${\cal C}_-^\de$. If $T$ is sufficiently
small the latter will be also true for $\Re u_0$. It follows that
\begin{figure}[t]
\begin{center}
\includegraphics[width=.7\textwidth]{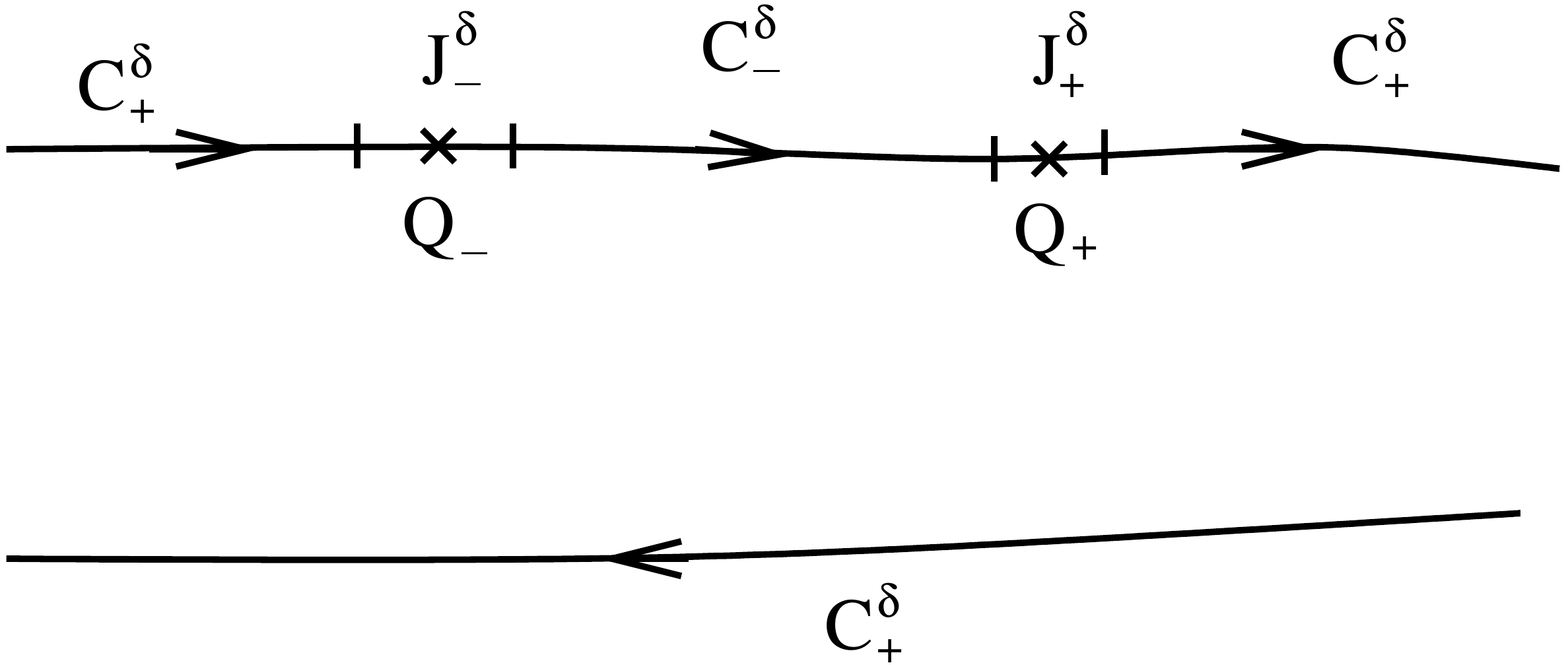}
\caption{\label{fig:contour_decomposition}
Deformation and decomposition of the contour $- ({\cal C}_{0,s} - \i \g/2)$.
}
\end{center}
\end{figure}%
\begin{equation}
     A = I_{{\cal C}_-^\de} + I_{J_-^\de} + I_{J_+^\de} + \CO (T^\infty) \epc
\end{equation}
where
\begin{equation}
     I_{\cal C} =
        \int_{\cal C} \frac{\rd \la}{2 \p \i T} \:
	   \biggl[ \frac{u'(\la)}{1 + \re^{u(\la)/T}} -
	           \frac{u_0'(\la)}{1 + \re^{u_0(\la)/T}} \biggr]
           L_{{\cal C}_{0,s} - \i \g/2} [z] (\la_+)
\end{equation}
and $\la_+$ is the boundary value from above.

Using Corollary~\ref{cor:zcauchy} we obtain
\begin{align} \label{icminus}
     I_{{\cal C}_-^\de} & =
        \int_{{\cal C}_-^\de} \frac{\rd \la}{2 \p \i T} \:
	   \bigl(u'(\la) - u_0'(\la)\bigr) L_{{\cal C}_{0,s} - \i \g/2} [z] (\la_+)
	   + \CO (T^\infty) \notag \\[1ex]
        & = - \int_{{\cal C}_-^\de} \frac{\rd \la}{2 \p \i} \: u_1' (\la)
              \int_{- Q}^Q \frac{\rd \m}{2 \p \i} \:
	         \cth(\m - \la_+) \uq_1 (\m) + \CO (T) \notag \\[1ex]
        & = - \int_{- Q}^Q \frac{\rd \la}{2 \p \i} \: u_1' (\la)
              \int_{- Q}^Q \frac{\rd \m}{2 \p \i} \:
	         \cth(\m - \la_+) \uq_1 (\m) + o(1) \notag \\[1ex]
        & = \int_{- Q}^Q \frac{\rd \la}{2 \p \i} 
            \int_{- Q}^Q \frac{\rd \m}{2 \p \i} \:
               \frac{\uq_1'(\la) \uq_1 (\m) - \uq_1'(\m) \uq_1 (\la)}
	            {2 \tgh (\la - \m)} \notag \\[1ex] & \mspace{216.mu}
               - \frac{1}{8 \p \i} \bigl( \uq_1^2 (Q) - \uq_1^2 (-Q) \bigr)
	       + o(1) \epc
\end{align}
where $o(1)$ denotes terms that vanish in the ordered limit
first $\de/T \rightarrow \infty$ and then $\de \rightarrow 0$.

In order to evaluate the remaining integrals we split them further. Setting
$\de u = (\uq - u_0)/T$ we shall write $I_{J_\pm^\de} = I_{J_\pm^\de}^{(1)}
+ I_{J_\pm^\de}^{(2)}$, where
\begin{subequations}
\begin{align}
     I_{J_\pm^\de}^{(1)} & =
        \int_{J_\pm^\de} \frac{\rd \la}{2 \p \i T} \: u' (\la)
	   \biggl[ \frac{1}{1 + \re^{\uq (\la)/T}} -
	           \frac{1}{1 + \re^{\uq (\la)/T - \de u(\la)}} \biggr]
           L_{{\cal C}_{0,s} - \i \g/2} [z] (\la_+) \epc \\[1ex]
     I_{J_\pm^\de}^{(2)} & =
        \int_{J_\pm^\de} \frac{\rd \la}{2 \p \i} \:
	   \frac{\de u' (\la)}{1 + \re^{\uq (\la)/T - \de u(\la)}}
           L_{{\cal C}_{0,s} - \i \g/2} [z] (\la_+) \epp
\end{align}
\end{subequations}

Recall that $\uq$ maps $J_\pm^\de$ to $[- \de, \de]$ in such a way that
$\uq (Q_\pm) = 0$. Hence, the substitution $t = \uq (\la)$ transforms
$I_{J_+^\de}$ into an integral of the form appearing on the left hand
side of (\ref{limint}). Since the requirements of Lemma~\ref{lem:limint}
are satisfied, we conclude that
\begin{align} \label{ijplus}
     I_{J_+^\de}^{(1)} & = \int_{- \infty}^\infty \frac{\rd t}{2 \p \i} \: 
	   \biggl[ \frac{1}{1 + \re^t} - \frac{1}{1 + \re^{t - \uq_1 (Q)}} \biggr]
	   \biggl\{ - \int_{- Q}^Q \frac{\rd \m}{2 \p \i} \:
	              \cth(\m - Q) \bigl( \uq_1 (\m) - \uq_1 (Q) \bigr) \notag \\[.5ex]
	   & \mspace{36.mu}
           - \frac{\uq_1 (Q)}{2 \p \i}
	     \ln \biggl( \frac{2 \p \i T}{\e'(Q) \sh (2Q)} \biggr)
           - \ln \biggl( \frac{\G (1/2 + t/2 \p \i)}
	                      {\G \bigl(1/2 + t/2 \p \i - \uq_1 (Q)/2 \p \i\bigr)}
			       \biggr) \biggr\} + o(1) \notag \\[1ex]
        & = \frac{\uq_1 (Q)}{2 \p \i}
	    \int_{- Q}^Q \frac{\rd \m}{2 \p \i} \:
	       \cth(\m - Q) \bigl( \uq_1 (\m) - \uq_1 (Q) \bigr)
           + \biggl( \frac{\uq_1 (Q)}{2 \p \i} \biggr)^2
	     \ln \biggl( \frac{2 \p \i T}{\e'(Q) \sh (2Q)} \biggr) \notag \\[.5ex]
	   & \mspace{180.mu}
           + \ln \biggl\{ G \biggl(1 - \frac{\uq_1 (Q)}{2 \p \i}\biggr)
	                  G \biggl(1 + \frac{\uq_1 (Q)}{2 \p \i}\biggr)
			  \biggr\} + o(1) \epp
\end{align}
Here we have used Corollary~\ref{cor:zcauchy} and dropped $\CO (T)$
corrections in the integrand in the first equation. The remaining
integrals have been calculated by means of residue calculus (cf.\
Appendix B of \cite{KMS11b}). A similar calculation yields
\begin{multline} \label{ijminus}
     I_{J_-^\de}^{(1)} = - \frac{\uq_1 (- Q)}{2 \p \i}
	    \int_{- Q}^Q \frac{\rd \m}{2 \p \i} \:
	       \cth(\m + Q) \bigl( \uq_1 (\m) - \uq_1 (- Q) \bigr) \\
           + \biggl( \frac{\uq_1 (- Q)}{2 \p \i} \biggr)^2
	     \ln \biggl( \frac{- 2 \p \i T}{\e'(Q) \sh (2Q)} \biggr) \\[.5ex]
           + \ln \biggl\{ G \biggl(1 + \frac{\uq_1 (- Q)}{2 \p \i}\biggr)
	                  G \biggl(1 - \frac{\uq_1 (- Q)}{2 \p \i}\biggr)
			  \biggr\} + o(1) \epp
\end{multline}

The remaining integrals will turn out to be $o(1)$. In order to
prove this assertion we split them once more into
$I_{J_\pm^\de}^{(2)} = I_{J_\pm^\de}^{(3)} + I_{J_\pm^\de}^{(4)}$, where
\begin{subequations}
\begin{align}
     I_{J_\pm^\de}^{(3)} & =
        \int_{J_\pm^\de} \frac{\rd \la}{2 \p \i} \:
	   \de u' (\la) \biggl[
	   \frac{1}{1 + \re^{\uq (\la)/T - \de u(\la)}}
	   - \Theta \biggl( - \frac{\uq (\la)}{T} \biggr) \biggr]
           L_{{\cal C}_{0,s} - \i \g/2} [z] (\la_+) \epc \\[1ex]
     I_{J_\pm^\de}^{(4)} & =
        \int_{J_\pm^\de} \frac{\rd \la}{2 \p \i} \:
	   \de u' (\la) \Theta \biggl( - \frac{\uq (\la)}{T} \biggr)
           L_{{\cal C}_{0,s} - \i \g/2} [z] (\la_+) \epp
\end{align}
\end{subequations}
Here $\Theta$ is the Heaviside step function. The integrals $I_{J_\pm^\de}^{(3)}$
can be treated as $I_{J_\pm^\de}^{(1)}$. Again Lemma~\ref{lem:limint} applies,
and the resulting terms are $o(1)$.

In order to estimate $I_{J_+^\de}^{(4)}$ we write it as
\begin{multline}
     I_{J_+^\de}^{(4)} =
        \int_{Q_+^-}^{Q_+} \frac{\rd \la}{2 \p \i} \: \uq_1' (\la)
	   \biggl\{ - \frac{\uq_1 (\la)}{2 \p \i} \ln \bigl( 2 \p \i T \bigr) \\
           - \ln \biggl( \frac{\G (1/2 + \uq (\la) /2 \p \i T)}
		         {\G \bigl(1/2 + \uq (\la) /2 \p \i T - \uq_1 (\la)/2 \p \i\bigr)}
		         \biggr) \biggr\} + o(1) \epp
\end{multline}
Here we have denoted the left end point of $J_+^\de$, where $\uq = - \de$, by
$Q_+^-$. We have inserted (\ref{closetoqpm}) in which we have suppressed the
$\CO (1)$ terms. Now the integral over the first term in curly brackets can
be evaluated explicitly. In order to estimate the second term we note that
due to Stirling's formula
\begin{equation}
    \ln \biggl( \frac{\G(1/2 + x)}{\G(1/2 + x - a)} \biggr) - a \ln (1 + x)
       + \frac{a + a^2/2}{1 + x} = \CO \bigl( 1/x^2 \bigr) \epp 
\end{equation}
Hence, if we rewrite $I_{J_+^\de}^{(4)}$ as
\begin{align}
     I_{J_+^\de}^{(4)} = & \frac{\uq_1^2 (Q_+^-) - \uq_1^2 (Q_+)}{8 (\p \i)^2}
                            \ln \bigl( 2 \p \i T \bigr) \notag \\[.5ex]
        & - \int_{Q_+^-}^{Q_+} \frac{\rd \la}{2 \p \i} \: \uq_1' (\la) \biggl\{
          \ln \biggl( \frac{\G (1/2 + \uq (\la) /2 \p \i T)}
		           {\G \bigl(1/2 + \uq (\la) /2 \p \i T
			      - \uq_1 (\la)/2 \p \i\bigr)} \biggr)
			    \notag \\ & \mspace{72.mu}
          - \frac{\uq_1 (\la)}{2 \p \i} \ln \biggl(1 + \frac{\uq (\la)}{2 \p \i T} \biggr)
	  + \frac{\uq_1 (\la)}{2 \p \i} \biggl( 1 + \frac{\uq_1 (\la)}{4 \p \i} \biggr)
	    \biggl[1 + \frac{\uq (\la)}{2 \p \i T} \biggr]^{- 1} \biggr\}
	    \notag \\[1ex] & \mspace{-50.mu}
        - \int_{Q_+^-}^{Q_+} \frac{\rd \la}{2 \p \i} \: \uq_1' (\la) \biggl\{
          \frac{\uq_1 (\la)}{2 \p \i} \ln \biggl(1 + \frac{\uq (\la)}{2 \p \i T} \biggr)
	  - \frac{\uq_1 (\la)}{2 \p \i} \biggl( 1 + \frac{\uq_1 (\la)}{4 \p \i} \biggr)
	    \biggl[1 + \frac{\uq (\la)}{2 \p \i T} \biggr]^{- 1} \biggr\} \notag \\
	    & \mspace{432.mu} + o(1) \epc
\end{align}
then, after substituting $x = \uq (\la)$, the integrand in the first integral
satisfies the requirements of Lemma~\ref{lem:limint}, and the integral turns
out to be $o(1)$. The second term in the second integral is $\CO (T)$ and
can be neglected. It follows that
\begin{multline}
     I_{J_+^\de}^{(4)} = \frac{\uq_1^2 (Q_+^-) - \uq_1^2 (Q_+)}{8 (\p \i)^2}
                            \ln \bigl( 2 \p \i T \bigr)
        - \int_{Q_+^-}^{Q_+} \frac{\rd \la}{2 \p \i} \: \uq_1' (\la)
          \frac{\uq_1 (\la)}{2 \p \i}
	  \ln \biggl(1 + \frac{\uq (\la)}{2 \p \i T} \biggr) + o(1) \\[1ex]
        = \frac{\uq_1^2 (Q_+^-) - \uq_1^2 (Q_+)}{8 (\p \i)^2} \ln ( - \de) + o(1)
	= o(1) \epp
\end{multline}
In order to estimate the remaining integral we performed a partial
integration and used that
\begin{equation}
     \int_{- \de}^0 \frac{\rd t}{2 \p \i} \: \frac{g (t) x}{1 + t x/ 2 \p \i}
        = - g(0) \ln \biggl( - \frac{x \de}{2 \p \i} \biggr) + o(1) \epc
\end{equation}
where $g \in {\cal C}^1 (I)$ for an open interval $I$ containing $0$, and
$o (1)$ denotes terms which vanish in the ordered limit first $x \de
\rightarrow \infty$ and then $\de \rightarrow 0$ (cf.\ Appendix~B of
\cite{KKMST09c}). The integral $I_{J_-^\de}^{(4)}$ can be treated in
a similar way and turns out to be $o(1)$ as well.

Summarizing the above we have shown that $A = I_{{\cal C}_-^\de} +
I_{J_-^\de}^{(1)} + I_{J_+^\de}^{(1)} + o(1)$. Then Lemma~\ref{lem:sing2}
follows from (\ref{icminus}), (\ref{ijplus}) and (\ref{ijminus}).
}

\bibliographystyle{amsplain}
\bibliography{hub,fab,andreas}

\end{document}